\documentclass[letterpaper,USenglish,cleveref,numberwithinsect,thm-restate]{no-lipics-v2022}
\usepackage[utf8]{inputenc}
\usepackage{color}
\usepackage{csquotes}
\usepackage{xspace}
\usepackage{amsmath}
\usepackage{amssymb}
\usepackage{mathtools}
\usepackage{microtype}
\usepackage{bbm}
\usepackage{bm}
\usepackage[linesnumbered, noend]{algorithm2e}
\usepackage{tikz}
\usetikzlibrary{arrows,arrows.meta,decorations.pathreplacing,decorations.pathmorphing,shapes,calc,patterns,shapes,matrix,math,quotes}
\usepackage{upgreek}

\nolinenumbers
\newcommand{\gap}{\bot}
\newcommand{\eps}{\varepsilon}

\newcommand{\ed}{\mathsf{ed}}
\def\dd{\mathinner{.\,.}}
\newcommand{\fragmentco}[2]{[#1\dd #2)}

\newcommand{\fragmentcc}[2]{[#1\dd #2]}
\newcommand{\range}[2]{\fragmentco{#1}{#2}}
\newcommand{\irange}[2]{\fragmentcc{#1}{#2}}
\newcommand{\gapcost}{g}
\newcommand{\mass}{\mu}
\DeclareMathOperator{\dist}{dist}

\newcommand{\Oh}{\ensuremath{\mathcal{O}}\xspace}
\newcommand{\tOh}{\ensuremath{\smash{\rlap{\raisebox{-0.2ex}{$\widetilde{\phantom{\Oh}}$}}\Oh}}\xspace}
\newcommand{\tO}{\tOh}

\newcommand{\ie}{i.e.\xspace}

\newcommand{\ceil}[1]{\lceil{#1}\rceil}
\newcommand{\set}[1]{\{#1\}}

\newcommand{\mycomment}[1]{}

\newtheorem{hypothesis}[theorem]{Hypothesis}
\crefname{hypothesis}{Hypothesis}{Hypotheses}

\title{Factor Three Approximation for Edit Distance}

\author{Egor Gorbachev}{ETH Zürich and Max Planck Institute for Informatics}{peltorator@pm.me}{https://orcid.org/0009-0005-5977-7986}{}

\authorrunning{E. Gorbachev}

\Copyright{Egor Gorbachev}

\acknowledgements{}

\begin{document}
\maketitle
\begin{abstract}
We give randomized algorithms for $3$-approximate edit distance in $\tO(N^{11/6})$ time for unweighted edit distance and in $\tO(N^{40/21})$ time for arbitrary metric edit weights, where $N$ is the total input length.\footnote{$\tO$ hides factors polynomial in $\log N$.}

For non-metric costs, we prove an unconditional $\Omega(N^2)$ oracle-query lower bound for every approximation factor depending only on $N$, even for symmetric weights or weights satisfying the triangle inequality (but not both).
Under the Orthogonal Vectors Hypothesis, we show a similar result for constant-size alphabets.
This holds even for symmetric weights over a size-$3$ alphabet or triangle-inequality weights over a size-$2$ alphabet.
In contrast, for symmetric weights over a binary alphabet we show an $\tO(N^{40/21})$-time $3$-approximation algorithm.
\end{abstract}

\section{Introduction}

The edit distance $\ed(X,Y)$ of two strings $X$ and $Y$ is the
smallest number of character insertions, deletions, and substitutions
needed to transform one string into the other.  For $X\in\Sigma^m$ and
$Y\in\Sigma^n$ over a finite alphabet $\Sigma$, we write
$N\coloneqq\max\{2,m+n\}$ and call $X$
the \emph{source string} and $Y$ the \emph{target string}.  The classical dynamic
program of Wagner and Fischer~\cite{WF} computes $\ed(X, Y)$ in
$\Oh(N^2)$ time by comparing all pairs of prefixes.  Subsequent work
obtained faster algorithms when the distance
$k\coloneqq\ed(X,Y)$ is small: Ukkonen~\cite{Ukk85} achieved time complexity $\Oh(Nk)$, while Landau and Vishkin~\cite{LV88}
gave an algorithm running in $\Oh(N+k^2)$ time.

Logarithmic improvements over the quadratic bound were also known,
starting with Masek and Paterson~\cite{MP80} for fixed alphabets.
Nevertheless, no \emph{truly subquadratic} exact algorithm, with
running time $\Oh(N^{2-\delta})$ for an absolute constant $\delta>0$,
is known for general inputs.
In 2015, Backurs and Indyk~\cite{BI15} gave a fine-grained
explanation for the lack of progress: such an algorithm would refute the Orthogonal Vectors Hypothesis (OVH) and thus also the Strong Exponential
Time Hypothesis (SETH).


As improvement for the exact computation is out of reach at the moment, it is natural to turn to developing approximation algorithms.
For example,
Andoni, Krauthgamer, and Onak~\cite{AKO10} obtained a
$(\log N)^{\Oh(1/\varepsilon)}$-approximation in $\Oh(N^{1+\varepsilon})$
time for each fixed $\varepsilon>0$.
The first constant-factor approximation in truly
subquadratic time was given by Chakraborty, Das, Goldenberg,
Kouck\'y, and Saks in 2018~\cite{CDGKS}: their original analysis
achieved factor 1680 in $\tO(N^{12/7})$ time.
Later in 2018, Andoni~\cite{Andoni18} improved the factor to
$3+\varepsilon$ in $\Oh_{\varepsilon}(N^{1.693})$ time.
Goldenberg, Rubinstein, and Saha~\cite{GRS21}, in their 2020 work,
obtained a $(3+\eps)$-approximation in $N^{1.6+o(1)}$ time.  In 2020, Andoni~\cite{Andoni20}
gave a simpler truly subquadratic $(3+\varepsilon)$-approximation algorithm.

Other results explore different parts of the tradeoff.
Andoni and Nosatzki~\cite{AN20} showed that, for every fixed
$\delta>0$, an $\Oh_{\delta}(1)$-factor approximation is possible in
$\Oh(N^{1+\delta})$ time.
More recently, Fox, Mader, Quanrud, Tavasoli, and Wang~\cite{FMQTW26}
obtained a $(3+\varepsilon)$-approximation in
$\tO_{\varepsilon}(N^{6/5}k^{2/5})$ time where $k$ is the edit distance $\ed(X, Y)$, improving
the time complexity for sublinear distances.
Mao and Rubinstein~\cite{MR26} obtained a $(1+\varepsilon)$-approximation
in $N^2/2^{\log^{\Omega(1)}N}$ time for every fixed $\varepsilon>0$.
This is a substantial saving, but it is still $N^{2-o(1)}$, rather than
truly subquadratic.  Thus, despite much progress in running time,
the general-input multiplicative guarantees in truly subquadratic
time stayed slightly above $3$ for nearly eight years.
We achieve factor three exactly, with no additive error.

\subparagraph{\boldmath Sketch of the $(3+\eps)$-approximation algorithms.}
The algorithms behind the $3+\eps$ bounds partition the source string $X$ into
consecutive substrings $U_i$, called \emph{source blocks}, of length at most $W = o(N)$ each.
The optimal (unknown) alignment induces a partition of $Y$ into target substrings $V_i$ such that the global alignment of $X$ onto $Y$ can be viewed as a concatenation of optimal alignments of $U_i$ onto $V_i$.
It remains to approximately compare each $U_i$ with all substrings of $Y$ and pick the appropriate $V_i$.
Computing all these comparisons directly would be too expensive.
Instead, we use the following routing idea:
We consider a small set $\mathcal C$ of intermediate strings we call \emph{centers} and compute distances from strings in $\mathcal{C}$ to all source blocks $U_i$ and all substrings of $Y$.
For each $U_i$, we pick a center $C_i$ minimizing $r_i \coloneqq \ed(U_i, C_i)$.
Then, whenever we need to approximate the distance $\ed(U_i, Y\fragmentco{p}{q})$ from a source block $U_i$ to some substring $Y\fragmentco{p}{q}$ we use triangle inequality of edit distance to claim that \[\ed(U_i, Y\fragmentco{p}{q}) \le r_i + \ed(C_i, Y\fragmentco{p}{q}) \le \ed(U_i, Y\fragmentco{p}{q}) + 2 r_i.\]
Therefore, if $\ed(U_i, Y\fragmentco{p}{q}) \ge r_i$, this routing leads to a $3$-approximation of $\ed(U_i, Y\fragmentco{p}{q})$.
A separate sample-based search supplies exact comparisons for
many of the inexpensive alignments with $\ed(U_i, V_i)<r_i$.
The search need not find every such comparison: it suffices to bound
the total cost of those it misses. This is where an excess above the factor three comes from.

\subparagraph{\boldmath How correcting a center saves the $\eps$ error.}
To achieve $3$-approximation, we counteract the $o(1)$ excess of sample-based search by marginally improving the triangle-inequality-based loss.
We move each selected center $C_i$ slightly closer to its
source block $U_i$.  Setting $C_i\coloneqq U_i$ would remove the routing
overhead entirely, but would also destroy the benefit of sharing
preprocessing.  It turns out that it is enough to undo \emph{one}
edit in an optimal alignment of $U_i$ onto $C_i$.
We choose the centers in such a way that $r_i \le |U_i| \le W$ holds.
The overhead can now be bounded by $2(r_i - 1) \le (2 - 2 / W) r_i$ instead of $2 r_i$.
We choose the remaining approximation errors small enough for this $\Theta(1 / W)$ saving to cover them.

The single corrected edit in the center $C_i$ is different for each source block $U_i$, but it does not
require separate preprocessing against the entire target string $Y$.
We use implicit DIST matrices~\cite{Schmidt98} to efficiently overcome this complication.

We are now ready to formulate our unit-cost approximation algorithm.
We say that an event holds \emph{with high probability} if, for any prescribed
constant $c>0$, its probability is at least $1-N^{-c}$ after choosing
the algorithm's constants appropriately.

\begin{restatable}[Unit costs]{theorem}{unitmain}\label{thm:unit}
There is a randomized algorithm that on every pair of input strings
$X\in\Sigma^m$ and $Y\in\Sigma^n$ runs in
$\tO((n + m)^{11/6})$ time and returns an alignment $\mathcal A$ of $X$ onto $Y$ of unweighted cost $\widehat D$ such that
$\widehat D\le3\ed(X,Y)$ holds with high probability.
\end{restatable}

\subsection{Metric-Weighted Edit Distance}

Assigning every edit the same cost is often too restrictive.
Different typing errors or biological mutations have different
likelihoods, and sequence comparison uses costs that distinguish
them; see~\cite{CKW23} for specific examples. \emph{Weighted} edit distance permits a separate
cost for each edit.

Let $\gap\notin\Sigma$ be a gap symbol and put
$\overline\Sigma\coloneqq\Sigma\cup\{\gap\}$.
A \emph{weight function} is an arbitrary function $w \colon \overline\Sigma^2\to\mathbb R_{\ge0}$ whose value is zero if and only if the input characters are equal.  Deleting $a$ costs $w(a,\gap)$,
inserting $a$ costs $w(\gap,a)$, and substituting $a$ by $b$ costs
$w(a,b)$.  An \emph{alignment} is a pair of equal-length strings
$A$ and $B$ over $\overline\Sigma$ that become $X$ and $Y$, respectively, when
their gaps are removed, and for no index $j$ have $A[j] = B[j] = \gap$.  Define
\[
 \ed^w(X,Y)\coloneqq
 \min_{(A,B)\text{ an alignment of }X,Y}\ \sum_j w(A[j],B[j]).
\]
The general definition imposes no symmetry or triangle inequality.
However, an important special case is when $w$ is a \emph{metric}: for all
$a,b,c\in\overline\Sigma$,
\[
 w(a,b)=0\ \Longleftrightarrow\ a=b,\qquad
 w(a,b)=w(b,a),\qquad
 w(a,c)\le w(a,b)+w(b,c).
\]
Then $\ed^w$ is itself a metric on strings.  In particular, the triangle
inequality and symmetry again permit routing comparisons through centers.

Approximation for arbitrary metric costs long lagged behind the
unit-cost case.  Kuszmaul~\cite{Kusz19} gave an
$\Oh(N^\delta)$-approximation in $\tO(N^{2-\delta})$ time for
each fixed $\delta\in(0,1)$.  Very recently, two independent works obtained the
first truly subquadratic constant-factor guarantees for arbitrary
metrics.  Mader, Tavasoli, and Wang~\cite{MTW26} gave a
$(3+\varepsilon)$-approximation in
$\tO(N^{7/4}/\varepsilon^8)$ time, and Das, Kipouridis, and
Kociumaka~\cite{DKK26} gave one in
$\tO(N^{8/5}/\varepsilon^{16/5})$ time.

These algorithms retain the center-based strategy but supply tools
needed for arbitrary metric weights.
We pair these metric tools with our center-correcting idea.
Undoing a single edit as in the unit-cost case, however, does not suffice anymore.
Instead, we replace some substring of the center by a short substring of the source block $U_i$.
The same DIST-matrix machinery from the unit-cost case allows us to efficiently reuse preprocessing for the centers.
Balancing the cost of these corrections with the sample-based search bounds gives our second result.

\begin{restatable}[Arbitrary metric weights]{theorem}{weightedmain}\label{thm:weighted}
There is a randomized algorithm that, given strings
$X\in\Sigma^m$ and $Y\in\Sigma^n$ and constant-time oracle access
to a metric weight function $w \colon \overline\Sigma^2\to\mathbb R_{\ge0}$,
runs in $\tO((n + m)^{40/21})$ time
and returns an alignment $\mathcal A$ of $X$ onto $Y$ of weighted cost $\widehat D$ such that
$\widehat D\le3\ed^w(X,Y)$ holds with high probability.
\end{restatable}

\subsection{What Fails if the Weight Function Is Not a Metric?}
\label{sec:intro-lowerbounds}

It is natural to ask why the above theorem is restricted to the case of \emph{metric} weight functions.
Can these approximation results extend to arbitrary weight functions?
Das, Kipouridis, and Kociumaka~\cite{DKK26} explicitly leave the nonmetric case unresolved.
We give two simple obstructions for arbitrary weight functions.
Each persists when we drop either symmetry or the triangle inequality.

\subparagraph{Oracle-based lower bound for large alphabets.}
If the alphabet has size $\Theta(N)$, its weight table has
$\Theta(N^2)$ entries.  Table size alone is not a lower bound:
an algorithm might not need to inspect most entries.
However, a cleverly designed instance may potentially force the algorithm to read a constant fraction of the weight table.
Huang~\cite{Huang88} used this idea to prove an unconditional quadratic lower bound for \emph{exact} weighted edit distance computation.
We present a modification of this result that rules out any approximation of weighted edit distance in $o(N^2)$ time.

Our strings make every one of quadratically many symbol pairs potentially
decisive.  In a baseline instance every alignment has cost at least
$B$ where $B$ is a parameter we can choose freely.
Changing any one of those pair costs creates an alignment of
cost $\Oh(N)$.
By choosing $B$ sufficiently large, any proposed
finite approximation factor must distinguish the baseline from
each alternative, forcing quadratically many queries.

\begin{restatable}[Oracle lower bound]{theorem}{nonmetricoraclemain}
\label{thm:nonmetric-oracle}
Let $\alpha \colon \mathbb Z_{> 0}\to\mathbb Z_{> 0}$ be an arbitrary
function. Every randomized algorithm that, given two strings $X$ and $Y$ and oracle access to an arbitrary (potentially non-metric) weight function $w \colon \overline\Sigma^2\to\mathbb R_{\ge0}$ returns a value $\widehat D$ satisfying
\[
 \frac{\ed^w(X,Y)}{\alpha(N)}\le\widehat D\le\alpha(N)\ed^w(X,Y)
\]
with probability at least $2/3$ requires $\Omega(N^2)$ oracle queries in
expectation on inputs of arbitrarily large total length $N$.
This holds even when all weights in $w$ belong to $\{0, 1,B\}$ for some positive integer $B$, and under either of two separate
restrictions: $w$ is symmetric, or $w$ satisfies the triangle
inequality, \ie, $w$ is a quasimetric.  (The two restrictions are
alternatives, not simultaneous assumptions.)
\end{restatable}

\subparagraph{Hardness for a constant-size alphabet.}
One may object that the above lower bound is artificial because it somehow forces the true input size to be quadratic.
To address this criticism, we give a lower bound for an alphabet of constant size.
In this case we do not have tools for developing unconditional lower bounds so we resort to a fine-grained conditional one.
Specifically, we base our lower bound on the Orthogonal Vectors
Hypothesis (OVH)~\cite{KW19}, the same hypothesis underlying the
fine-grained lower bound for exact edit distance~\cite{BI15}.
OVH asserts that for every
$\varepsilon>0$ there is a constant $c>0$ such that no
$\Oh(n^{2-\varepsilon})$-time algorithm can decide whether two lists
of $n$ vectors in $\{0,1\}^{\lceil c\log_2 n\rceil}$ contain a pair
whose inner product is zero.  Its randomized version excludes
algorithms succeeding with probability at least $2/3$ on every
input.

\begin{restatable}[Three-symbol symmetric hardness]{theorem}{nonmetricfixedmain}
\label{thm:nonmetric-fixed}
Let $\alpha \colon \mathbb Z_{>0}\to\mathbb Z_{>0}$ be an arbitrary
function computable in $N^{1+o(1)}$ time.  For every fixed $\delta>0$,
assuming the (randomized) Orthogonal Vectors Hypothesis, there is no
(randomized) algorithm with worst-case running time $\Oh(N^{2-\delta})$
that, given strings $X$ and $Y$ of total length $N$ and oracle access to a
symmetric weight function $w \colon \overline\Sigma^2\to\mathbb R_{\ge0}$,
returns a value $\widehat D$ satisfying
\[
 \frac{\ed^w(X,Y)}{\alpha(N)}\le\widehat D\le\alpha(N)\ed^w(X,Y)
\]
(with probability at least $2/3$).
This holds even when $|\Sigma|=3$ and all weights belong to
$\{0,1,B\}$ for some positive integer $B$.
\end{restatable}

Interestingly, the alphabet size in this lower bound is tight.  Over a binary
alphabet, arbitrary symmetric weights reduce exactly to metric weights
plus a nonnegative correction computed from character counts.
Consequently, our metric approximation also handles this case.

\begin{restatable}[Symmetric weights for a binary alphabet]{theorem}{binarysymmetricmain}
\label{thm:binary-symmetric}
There is a randomized algorithm that, given strings $X\in\{0, 1\}^m$ and $Y\in\{0,1\}^n$
and oracle access to an arbitrary symmetric weight function
$w \colon \{0,1,\gap\}^2\to\mathbb R_{\ge0}$, runs in
$\tO((n+m)^{40/21})$ time and returns an alignment $\mathcal A$ of
$X$ onto $Y$ of cost $\widehat D$ such that $\widehat D\le3\ed^w(X,Y)$
holds with high probability.  No triangle inequality is required.
\end{restatable}

Without symmetry, hardness already occurs over a binary alphabet,
even when the weights satisfy the triangle inequality.

\begin{restatable}[Binary alphabet quasimetric hardness]{theorem}{binaryquasimetrichardmain}
\label{thm:binary-quasimetric-hardness}
Let $\alpha \colon \mathbb Z_{>0}\to\mathbb Z_{>0}$ be an arbitrary
function computable in $N^{1+o(1)}$ time.  For every fixed $\delta>0$,
assuming the (randomized) Orthogonal Vectors Hypothesis, there is no
(randomized) algorithm with worst-case running time $\Oh(N^{2-\delta})$
that, given strings $X,Y\in\{0,1\}^*$ of total length $N$ and oracle
access to a weight function $w \colon \{0,1,\gap\}^2\to\mathbb R_{\ge0}$
satisfying the triangle inequality, returns a value $\widehat D$ satisfying
\[
 \frac{\ed^w(X,Y)}{\alpha(N)}\le\widehat D\le\alpha(N)\ed^w(X,Y)
\]
(with probability at least $2/3$).
This holds even when all weights belong to $\{0,1,B\}$ for some
positive integer $B$.
\end{restatable}

These lower bounds show that both assumptions on the weight functions
in \cref{thm:weighted} are essential in general.

\subsection{Open Questions}

The tradeoff between running time and approximation remains far from
understood, even for ordinary edit distance.  We highlight the
following directions.

\begin{enumerate}
\item \textbf{What is the best approximation factor for truly subquadratic algorithms?}
Can one obtain a $(3-\varepsilon)$-approximation in truly subquadratic
time for some absolute constant $\varepsilon>0$?  Our correction
pays for losses that vanish with the input length.  Saving a fixed
fraction of the total cost would require new ideas.
The approximation scheme of Mao and Rubinstein~\cite{MR26} makes a step in this direction from a different starting point but their time complexity of $N^2/2^{\log^{\Omega(1)}N}$ is not truly subquadratic.
More ambitiously, can one
obtain a $(1+\varepsilon)$-approximation in truly subquadratic time
for every fixed $\varepsilon>0$?  No known fine-grained lower bound
even excludes $(1+\varepsilon)$-approximation in
$\Oh(N)$ time for any fixed positive $\varepsilon$.

\item \textbf{Faster factor-three algorithms.}
Can the exponents $11/6$ and $40/21$ be reduced?  A natural target
is to approach the $N^{8/5+o(1)}$ running times available for
factors just above three, for both unit and arbitrary metric costs.
It is also unclear whether arbitrary metrics must incur any
additional polynomial overhead over the unit-cost case.

\item \textbf{Constant factor in almost-linear time.}
Does there exist an absolute constant $C$ such that $C$-approximation of edit distance can be achieved in
$N^{1+o(1)}$ time?  Andoni--Nosatzki~\cite{AN20} allow every fixed
exponent $1+\delta$, but their approximation factor depends on $\delta$.
For arbitrary metric weights, even an analogue
of their fixed-$\delta$ tradeoff would be valuable.

\item \textbf{Deterministic constant-factor approximation.}
Can the sampling in these algorithms be replaced so as to obtain
a \emph{deterministic} constant-factor approximation in truly subquadratic
time?  This remains a natural question even for unit costs.
For the related longest common subsequence problem,
Boneh, Golan, and Kraus~\cite{BGK25LCS} give a deterministic
$\Oh(N^{3/4}\log N)$-approximation in $\tO(N)$ time over arbitrary
alphabets.

\item \textbf{Approximation schemes for arbitrary metrics.}
Can the $(1+\varepsilon)$-approximation of Mao--Rubinstein~\cite{MR26}
be extended to arbitrary metric weights with a saving of the form
$2^{\log^{\Omega(1)}N}$ over quadratic time?  Such a result would address a
different tradeoff from the truly subquadratic factor-three
algorithms considered here.  Logarithmic savings are known for exact
weighted edit distance over a constant-size alphabet: Crochemore,
Landau, and Ziv-Ukelson~\cite{CLZ03} achieve $\Oh(N^2/\log N)$ time
with unrestricted real-valued scoring matrices.  This does not give
an alphabet-independent speedup in our metric-oracle model.
A weaker target is a $(1+\varepsilon)$-approximation in
$\Oh(N^2/\log^c N)$ time for some constant $c>0$ and every fixed
$\varepsilon>0$, independently of the alphabet size and the
numerical range of the costs.
\end{enumerate}

\subsection{Paper Organization}

\cref{sec:preliminaries} introduces string notation,
alignment graphs, and the DIST-matrix and matrix-searching
tools used in both algorithms.
\cref{sec:overview-unit} proves our unit-cost algorithm while \cref{sec:overview-weighted} proves its weighted counterpart.
\cref{sec:nonmetric} proves the unconditional oracle lower
bound, the small-alphabet conditional lower bounds, and the approximation
algorithm for symmetric weight functions over a binary alphabet.

\section{Preliminaries}
\label{sec:preliminaries}

We collect some notation and standard algorithmic tools.

\subparagraph{Strings and edit distance.}
\label{sec:prelim-strings}
Write $m\coloneqq |X|$, $n\coloneqq |Y|$, and $N\coloneqq \max\{2,m+n\}$.
The empty string is denoted by $\eps$.  For integers $a\le b$, write $\irange{a}{b}$ for the
set of integers from $a$ through $b$, inclusive, and $\range{a}{b}$ when $b$ is excluded.
We index string characters from zero: $S[i]$ is the character at position
$i\in\range{0}{|S|}$.  For $a, b \in \fragmentcc{0}{|S|}$ with $a \le b$, let
$S\range{a}{b}$ be the substring of $S$ containing the characters with indices in $\range{a}{b}$, with
$S\range{a}{a}=\eps$.
Substrings at different positions remain different positioned objects
even when their contents coincide.

We write $\ed(X,Y)$ for unit-cost edit distance and $\ed^w(X,Y)$
for the minimum alignment cost under a weight function $w$.
Unit-cost edit distance is a metric on finite strings, including the
empty string: for all strings $F,G,H$,
\[
 \ed(F,G)=0\ \Longleftrightarrow\ F=G,\qquad
 \ed(F,G)=\ed(G,F),\qquad
 \ed(F,H)\le\ed(F,G)+\ed(G,H).
\]
If $w$ is a metric on $\overline\Sigma\coloneqq\Sigma\cup\{\gap\}$,
then $\ed^w$ is also a metric on strings and has the same properties.

\subparagraph{Alignment graph.}
\label{sec:prelim-alignment}
The alignment graph has a vertex $(x,y)$ for
every $x\in\irange{0}{m}$ and $y\in\irange{0}{n}$, representing the
prefixes $X\range{0}{x}$ and $Y\range{0}{y}$.
Its horizontal edge $(x,y)\to(x+1,y)$ represents the deletion of $X[x]$ and costs
$w(X[x],\gap)$; its vertical edge $(x,y)\to(x,y+1)$ represents the insertion of $Y[y]$
and costs $w(\gap,Y[y])$.  A diagonal edge
$(x,y)\to(x+1,y+1)$ pairs the two symbols $X[x]$ and $Y[y]$ for a cost of $w(X[x],Y[y])$.
Every path from $(0,0)$ to $(m,n)$ in the alignment graph
describes an alignment of the same cost.  The usual dynamic program
finds a shortest path and the corresponding alignment in $\Oh(mn)$ time~\cite{WF}.
In our drawings, $(0,0)$ is at the top left: source coordinates
increase to the right and target coordinates increase downward.
\cref{fig:unit-alignment} illustrates the unit-cost case.

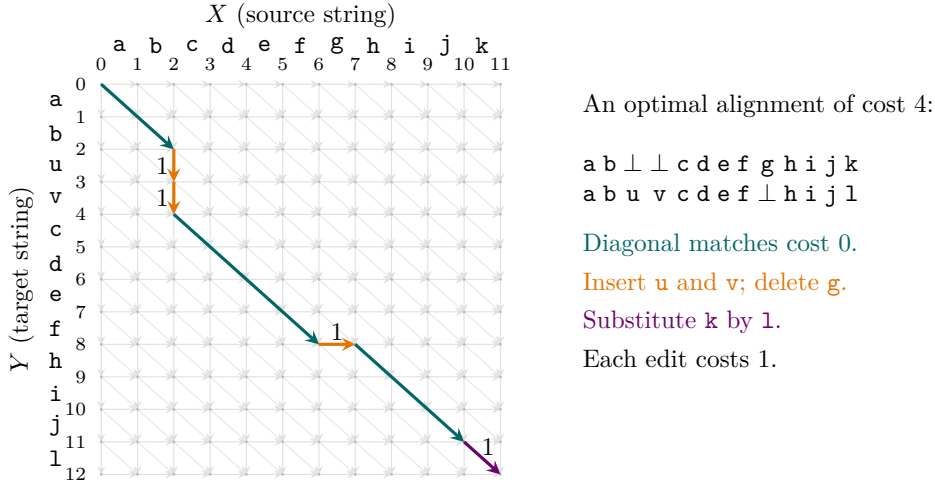
\begin{figure}
\centering
\begin{tikzpicture}[x=0.48cm,y=-0.43cm,>=stealth,font=\small]
  \foreach \x in {0,...,11} {
    \foreach \y in {0,...,11} {\draw[gray!25,->] (\x,\y)--(\x,\y+1);}}
  \foreach \x in {0,...,10} {
    \foreach \y in {0,...,12} {\draw[gray!25,->] (\x,\y)--(\x+1,\y);}
    \foreach \y in {0,...,11} {\draw[gray!20,->] (\x,\y)--(\x+1,\y+1);}}
  \foreach \x in {0,...,11} {
    \node[above] at (\x,-0.15) {\scriptsize $\x$};
    \foreach \y in {0,...,12} {\fill[gray!50] (\x,\y) circle (0.55pt);}}
  \foreach \y in {0,...,12} {\node[left] at (-0.15,\y) {\scriptsize $\y$};}
  \foreach \x/\a in {0.5/a,1.5/b,2.5/c,3.5/d,4.5/e,5.5/f,6.5/g,7.5/h,8.5/i,9.5/j,10.5/k}
    {\node at (\x,-1.2) {\texttt{\a}};}
  \foreach \y/\a in {0.5/a,1.5/b,2.5/u,3.5/v,4.5/c,5.5/d,6.5/e,7.5/f,8.5/h,9.5/i,10.5/j,11.5/l}
    {\node at (-1.25,\y) {\texttt{\a}};}
  \node at (5.5,-2.1) {$X$ (source string)};
  \node[rotate=90] at (-2.2,6) {$Y$ (target string)};
  \draw[very thick,teal!80!black,->] (0,0)--(1,1)--(2,2);
  \draw[very thick,orange!90!black,->] (2,2)--(2,3);
  \draw[very thick,orange!90!black,->] (2,3)--(2,4);
  \draw[very thick,teal!80!black,->] (2,4)--(3,5)--(4,6)--(5,7)--(6,8);
  \draw[very thick,orange!90!black,->] (6,8)--(7,8);
  \draw[very thick,teal!80!black,->] (7,8)--(8,9)--(9,10)--(10,11);
  \draw[very thick,violet!85!black,->] (10,11)--(11,12);
  \node[left,inner sep=1pt] at (1.95,2.5) {$1$};
  \node[left,inner sep=1pt] at (1.95,3.5) {$1$};
  \node[above,inner sep=1pt] at (6.5,7.95) {$1$};
  \node[right,inner sep=1pt] at (10.4,11.15) {$1$};
  \node[anchor=north west,align=left] at (13,0)
    {An optimal alignment of cost $4$:};
  \node[anchor=north west] at (13,1.7)
    {$\begin{array}{@{}*{13}{c@{\;}}}
       \mathtt a&\mathtt b&\gap&\gap&\mathtt c&\mathtt d&\mathtt e&\mathtt f&\mathtt g&\mathtt h&\mathtt i&\mathtt j&\mathtt k\\
       \mathtt a&\mathtt b&\mathtt u&\mathtt v&\mathtt c&\mathtt d&\mathtt e&\mathtt f&\gap&\mathtt h&\mathtt i&\mathtt j&\mathtt l
       \end{array}$};
  \node[anchor=north west,align=left,text width=5.2cm] at (13,4.3)
    {\textcolor{teal!80!black}{Diagonal matches cost $0$.}\\[3pt]
     \textcolor{orange!90!black}{Insert \texttt{u} and \texttt{v}; delete \texttt{g}.}\\[3pt]
     \textcolor{violet!85!black}{Substitute \texttt{k} by \texttt{l}.}\\[3pt]
     Each edit costs $1$.};
\end{tikzpicture}
\caption{The edit-distance alignment graph for $X = \mathtt{abcdefghijk}$ and
$Y = \mathtt{abuvcdefhijl}$.
The highlighted path is an optimal alignment of cost $4$.}
\label{fig:unit-alignment}
\end{figure}

\subparagraph{DIST matrices.}
\label{sec:prelim-oracles}
For strings $F,J$, we use the left-to-right part of their \emph{DIST matrix},
whose entries are
\[
 \operatorname{DIST}_{F,J}^w[q, p]\coloneqq\ed^w(F,J\range{p}{q})
 =\dist\bigl((0,p),(|F|,q)\bigr)
 \qquad(0\le p\le q\le|J|).
\]
Set the entry to $+\infty$ when $p>q$.  We represent this matrix
implicitly, without enumerating all its entries.

\begin{theorem}[Implicit weighted DIST matrices, {\cite{Schmidt98}}]
\label{thm:dist}
For strings $F,J$ and arbitrary finite nonnegative real edit costs,
an implicit representation of their DIST matrix can be constructed in
$\tO((|F|+1)(|J|+1))$ deterministic time and space.
It answers exact entry queries in $\tO(1)$ time.
\end{theorem}

\subparagraph{Totally monotone matrices and the SMAWK algorithm.}
\label{sec:prelim-boundary}
A matrix $A$ is called \emph{totally monotone} if in each submatrix (not necessarily contiguous), the column indices of the leftmost row minima are nondecreasing.
In particular, DIST matrices are totally monotone~\cite[Section 3, equation (2)]{HLLW}.
This property also holds after adding an arbitrary finite constant to each column of these matrices.
The celebrated SMAWK algorithm~\cite{SMAWK} allows searching for row minima of a totally monotone matrix much faster than the time it takes to examine all matrix entries.

\begin{theorem}[SMAWK, {\cite{SMAWK}}]
\label{thm:prelim-smawk}
Given oracle access to a totally monotone matrix with $r$ rows and
$c$ columns, all row minima and a minimizing column for each row
can be found using $\Oh(r+c)$ matrix entry queries and $\Oh(r+c)$ additional
time.
\end{theorem}

\section{The Unit-Cost Algorithm}
\label{sec:overview-unit}

We start by proving the unit-cost result.  The construction is close
to earlier unweighted and weighted approximation algorithms, but combines
ideas from the sampled-center and hierarchical-search framework of
Andoni~\cite{Andoni18,Andoni20} and its weighted counterparts by
Mader, Tavasoli, and Wang~\cite{MTW26} and Das, Kipouridis, and
Kociumaka~\cite{DKK26}.  We give a self-contained analysis specialized to unit costs.
Our main novel idea is small: undo one edit in each source block's alignment
to its center.  This saves enough to absorb the sampling loss above the approximation factor three.
We introduce some minor additional simplifications to make the ideas from the literature arguably even simpler.

\subsection{Proof Overview}
\label{sec:unit-overview-graph}

The alignment graph of \cref{sec:preliminaries} gives an
exact quadratic-time algorithm.
Note that if we remove all diagonal edges (matches and substitutions) from the alignment graph, computing distances in the resulting graph becomes trivial:
the distance from vertex $(a, p)$ to vertex $(b, q)$ with $a \le b$ and $p \le q$ is equal to $(b - a) + (q - p)$.
However, such distance may not approximate the corresponding distance in the alignment graph well enough if in the alignment graph the corresponding shortest path uses a lot of weight-zero diagonal edges (matches).
To circumvent this obstruction, we augment this trivial graph with additional edges that help preserve the distances from the alignment graph approximately while ensuring that the distances in the new graph can be computed in truly subquadratic time.
For some pairs of vertices $(a, p)$ and $(b, q)$ we add an edge $(a, p) \to (b, q)$ of weight $z$ satisfying $z\ge\ed(X\range{a}{b},Y\range{p}{q})$.
We call such edges \emph{certifying shortcuts} (or just \emph{shortcuts} for short) and the resulting graph the \emph{shortcut graph}.
Every path made of certified shortcuts and ordinary gap edges (insertions and deletions)
expands into an alignment of no greater cost: replace each shortcut
by an optimal alignment of its two substrings and concatenate.
Thus a shortest-path computation using these edges always gives
an upper bound on the edit distance.
Our task is to add a sufficient number of shortcuts to ensure that the distances in the shortcut graph do not exceed the corresponding distances in the alignment graph by too much.

To avoid processing every column of the alignment graph, choose a maximum block length $W$ and partition
$X$ into nonempty consecutive \emph{source blocks}
$U_i\coloneqq X\range{x_i}{x_{i+1}}$, $i\in\fragmentco{0}{R}$, where
$x_0=0$, $x_R=m$, and every block has length $W$ except possibly
the last. 
In \cref{fig:unit-blocks}, each source block occupies a vertical strip $[x_i,x_{i+1}]\times[0,n]$.
Split an optimal alignment at its first vertex on each interior source block boundary.
Its pieces align $U_i$ with successive, possibly empty substrings
$V_i\coloneqq Y\range{p_i}{p_{i+1}}$, with $p_0=0$ and $p_R=n$.
Write $d \coloneqq \ed(X, Y)$ and let $d_i = \ed(U_i,V_i)$ be the cost of the alignment piece $i$.
We have $\sum_i d_i=d$.

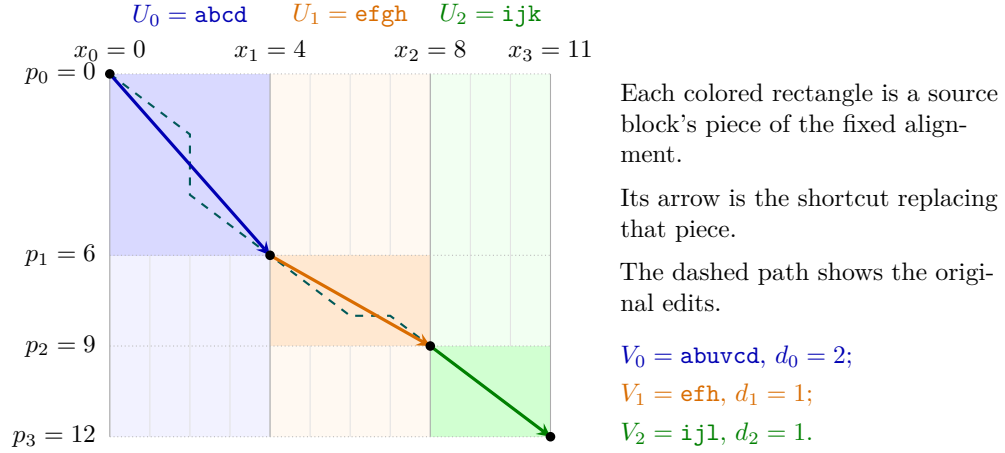
\begin{figure}
\centering
\begin{tikzpicture}[x=0.53cm,y=-0.40cm,>=stealth,font=\small]
  \fill[blue!5] (0,0) rectangle (4,12);
  \fill[orange!5] (4,0) rectangle (8,12);
  \fill[green!5] (8,0) rectangle (11,12);
  \draw[step=1,gray!22] (0,0) grid (11,12);
  \fill[blue!15] (0,0) rectangle (4,6);
  \fill[orange!18] (4,6) rectangle (8,9);
  \fill[green!18] (8,9) rectangle (11,12);
  \foreach \x in {0,4,8,11} {\draw[gray!65] (\x,0)--(\x,12);}
  \foreach \y in {0,6,9,12} {\draw[gray!55,densely dotted] (0,\y)--(11,\y);}
  \draw[thick,teal!70!black,dashed]
    (0,0)--(1,1)--(2,2)--(2,3)--(2,4)--(3,5)--(4,6)
    --(5,7)--(6,8)--(7,8)--(8,9)--(9,10)--(10,11)--(11,12);
  \draw[very thick,blue!70!black,->] (0,0)--(4,6);
  \draw[very thick,orange!85!black,->] (4,6)--(8,9);
  \draw[very thick,green!50!black,->] (8,9)--(11,12);
  \foreach \x/\y in {0/0,4/6,8/9,11/12}
    {\fill (\x,\y) circle (1.8pt);}
  \foreach \x/\lab in {0/{x_0=0},4/{x_1=4},8/{x_2=8},11/{x_3=11}}
    {\node[above] at (\x,-0.1) {$\lab$};}
  \node[above,blue!70!black] at (2,-1.25) {$U_0=\mathtt{abcd}$};
  \node[above,orange!85!black] at (6,-1.25) {$U_1=\mathtt{efgh}$};
  \node[above,green!50!black] at (9.5,-1.25) {$U_2=\mathtt{ijk}$};
  \foreach \y/\lab in {0/{p_0=0},6/{p_1=6},9/{p_2=9},12/{p_3=12}}
    {\node[left] at (-0.1,\y) {$\lab$};}
  \node[anchor=north west,align=left,text width=5cm] at (12.5,0)
    {Each colored rectangle is a source block's piece of the fixed alignment.\\[5pt]
     Its arrow is the shortcut replacing that piece.\\[5pt]
     The dashed path shows the original edits.\\[8pt]
     \textcolor{blue!70!black}{$V_0=\mathtt{abuvcd}$, $d_0=2$;}\\[3pt]
     \textcolor{orange!85!black}{$V_1=\mathtt{efh}$, $d_1=1$;}\\[3pt]
     \textcolor{green!50!black}{$V_2=\mathtt{ijl}$, $d_2=1$.}};
\end{tikzpicture}
\caption{Partitioning the source string in \cref{fig:unit-alignment}
into source blocks of maximum length $W = 4$.}
\label{fig:unit-blocks}
\end{figure}

Instead of running the DP inside each full strip, use certified
shortcuts from $(x_i,p)$ to $(x_{i+1},q)$, with $p\le q$.
If every such shortcut had
its exact label, the shortest path would equal $d$.
Our goal is to add some shortcuts such that $\ed(U_i, V_i)$ is well-approximated.

\subparagraph{How the previous approximation framework works.}
Pick some substrings of $Y$ as reference strings, called \emph{centers},
and assign each source block its closest center $C_i$.  Put
$r_i\coloneqq \ed(U_i,C_i)$, the source block's \emph{radius}.
The triangle inequality shows that routing a comparison through
a center gives, for every string $T$,
\begin{equation}\label{eq:overview-center-routing}
 \ed(U_i,T)\le r_i+\ed(C_i,T)
             \le\ed(U_i,T)+2r_i.
\end{equation}
If $r_i\le\ed(U_i,T)$, this is a factor-three approximation.  Only
comparisons closer than the center require further attention.
Random centers make such comparisons sparse, and a second sampling
procedure searches for them. A shortest-path computation over the resulting shortcuts gives
the $3+\varepsilon$ guarantee~\cite{Andoni18,Andoni20}.

\subparagraph{The additional operation.}
We offset the $+\eps$ loss by moving each center one edit closer to its source block.
Allowing the empty center ensures $r_i\le |U_i|\le W$.
If $r_i>0$, undo any one edit in an optimal alignment
of $U_i$ onto $C_i$.  The corrected center $\widehat C_i$ has distance
$\widehat r_i\coloneqq r_i-1 \le (1 - 1 / W) r_i$ to $U_i$.  Thus the routing overhead in
\cref{eq:overview-center-routing} decreases from $2r_i$ to $(2 - 2 / W) r_i$.
The rest of the construction selects suitable comparison,
sampling, and boundary-distance tools so that this saving counteracts the $+\eps$ in the approximation guarantee while making sure that the algorithm can still be implemented efficiently.

\subsection{Sampling Centers and Undoing One Edit}
\label{sec:unit-overview-centers}

We now formally describe our algorithm.
The algorithm works in prescribed time complexity with high probability.
If the algorithm exceeds the desired time bound, we terminate the algorithm and return a trivial alignment of $X$ onto $Y$ that deletes all characters of $X$ and inserts all characters of $Y$.

Choose an integer $W \in \fragmentcc{1}{N}$ (eventually about $N^{1/6}$).
If either input string has length below $W$, ordinary
DP takes $\Oh(NW)$ time.
Furthermore, if $X = Y$, return the trivial alignment of weight $0$.
Hence assume $m,n\ge W$ and $d \ge 1$.
There are $R\coloneqq \lceil m/W\rceil$ source blocks $U_i$.

Choose an integer $k \in \fragmentcc{1}{N}$ (eventually about $N^{1/3}$).
Independently and uniformly sample $\tO(N/k)$ target string starts from $\fragmentco{0}{n}$.
A sampled start $p$ supplies the segment
\[
 B_p\coloneqq Y\range{p}{\min\{n,p+2W\}}
\]
and \emph{all} its prefixes as candidate centers.  For each source block $U_i$ and
sampled start $p$, run ordinary DP against $B_p$.  The final DP column gives the distances to all the prefixes of $B_p$ at once.
Out of all such prefixes over all sampled starts $p$, choose a minimum-distance candidate substring $C_i$ and retain its optimal alignment and radius $r_i \coloneqq \ed(U_i, C_i)$.
Note that an empty center is allowed so we have a guarantee $r_i\le \ed(U_i, \eps) = |U_i| \le W$.
Furthermore, $\ed(J, F) \ge ||J| - |F||$ holds for arbitrary strings $J$ and $F$ so bounding the length of $B_p$ by $2W$ does not lose any potential longer center.

For a fixed source block $U_i$, rank the $n$ starting positions of $Y$ by their best prefix distance to $U_i$.
A draw hits at least one of the first $k$ starts with probability $k/n$; $t$ independent
draws miss them with probability $(1-k/n)^t\le e^{-tk/n}$.
A logarithmic multiple of $n/k$ draws therefore hits them with high
probability.
The chosen radius is therefore
no larger than the $k$-th ranked value.
By the union bound, we may assume that 
for every $i$, at most $k$ starts $p$ admit a $q\in \fragmentcc{p}{n}$
with $\ed(U_i,Y\range{p}{q})<r_i$.
The argument above assumes $k \le n$. The case $k > n$ follows trivially.

\subparagraph{The corrected centers.}
We now describe how we correct the centers $C_i$ by undoing a single edit in the optimal alignment of $U_i$ onto $C_i$.
If $r_i=0$, set $\widehat C_i\coloneqq C_i$.
Otherwise, select an arbitrary positive-cost column of the recorded optimal
alignment of $U_i$ onto $C_i$.  For a substitution, replace that center symbol by its
source partner.  For an inserted center symbol, delete it.  For a
deleted source symbol, insert it at its place in the center.
Call the resulting string $\widehat C_i$.
We have
\begin{equation}\label{eq:overview-one-edit}
 \widehat r_i\coloneqq \ed(U_i,\widehat C_i)
       =\max\{0,r_i-1\}.
\end{equation}
We use the corrected center to add certified shortcuts from $(x_i, p)$ to $(x_{i+1}, q)$ for every $p, q \in \fragmentcc{0}{n}$ with $p \le q$ of weight
\begin{equation}\label{eq:overview-repaired-routing}
 z_i(p,q)\coloneqq \widehat r_i+\ed(\widehat C_i,Y\range{p}{q}),\qquad
 \ed(U_i,Y\range{p}{q})\le z_i(p,q)
 \le\ed(U_i,Y\range{p}{q})+2\widehat r_i.
\end{equation}
We call these shortcuts \emph{routing shortcuts}.
Note that there are $\Oh(n^2)$ such shortcuts for each $U_i$ so we cannot construct them explicitly.
We later explain how we nevertheless manage to compute distances through them efficiently.

\subsection{Discovering the Exact Shortcuts}
\label{sec:unit-overview-search}

The routing shortcuts suffice when the alignment is expensive relative to
the radii.  For example, if over any consecutive group of source blocks the sum of
radii is at most the corresponding optimal alignment cost.
It therefore remains to search for better shortcuts in regions whose alignment
is inexpensive at every scale.

The \emph{diagonal offset} of a vertex $(x,y)$ is $y-x$.
Along an alignment, this offset changes only on gap edges.  Consequently,
an alignment graph path from $(a,p)$ to $(b,q)$ of cost $c$ satisfies
\begin{equation}\label{eq:unit-offset}
 |(q-b)-(p-a)|\le c.
\end{equation}
Therefore, if the optimal alignment is inexpensive in the region around $(x_h, p_h)$, finding this vertex may guide us to find vertices $(x_i, p_i)$ for $i$ close to $h$.
We use this idea in the following \emph{exact shortcut} algorithm.

Build a balanced binary tree on the ordered source block indices $\fragmentco{0}{R}$.  For a node
$v$ of the binary tree, let
\[
 \rho_v\coloneqq \sum_{i\in v} r_i.
\]
Choose an accuracy parameter $\eta\in(0,1)$ (eventually about $N^{-1 / 6}$).
At every node independently sample $\tO(1/\eta)$ source blocks uniformly
among its blocks with positive radius $r_i > 0$.
Skip nodes containing no positive-radius blocks.
This tree search is based on Andoni~\cite[Sections 3.1--3.2]{Andoni18}.
We then test the sampled source blocks against some substrings of $Y$.

An \emph{anchor} is a diagonal offset $\Delta$ obtained from a sampled source block
$h$ and a target start $p$, where $\Delta\coloneqq p-x_h$, for which at least
one tested substring $Y\range{p}{q}$ satisfies $\ed(U_h,Y\range{p}{q})<r_h$.

\begin{enumerate}
\item \textbf{At the root,} for each sampled source block $i$, test every start
    $p\in\fragmentco{0}{n}$ and every end $q\in\fragmentcc{p}{\min\{n,p+2W\}}$.
Accept exactly the substrings of distance below $r_i$ and store their
anchors.
\item \textbf{\boldmath At a child $v$ of $u$,} for each sampled source block $i$ and
each parent anchor $\Delta$, search the integer starts in
\[
 \irange{x_i+\Delta-\rho_u}{x_i+\Delta+\rho_u}\cap\fragmentco{0}{n}.
\]
Test all ends up to
$2W$ characters after each searched start, and keep the accepted anchors.
\item \textbf{\boldmath At a leaf for source block $i$,} output an exact shortcut $(x_i, p) \to (x_{i+1}, q)$ for
    every accepted substring $Y\fragmentco{p}{q}$.
\end{enumerate}

Use the implicit DIST matrices (\cref{thm:dist}) to test each sampled source block against the whole interval of target starts at once.  If the interval contains
$\ell$ starts, extend its target segment by up to $2W$ characters
to the right.  Preprocess the DIST matrix of $U_i$ against that
segment and answer the $\Oh(W\ell)$ tests in
$\tO(W(\ell+W))$ time in total.

\subsection{Evaluating the Distance in the Shortcut Graph and Recovering the Alignment}
\label{sec:unit-overview-evaluation}

Our shortcut graph will consist of gap edges, routing shortcuts, and exact shortcuts.
Exact shortcuts are constructed explicitly while gap edges and routing shortcuts are processed in batches.
Process the source blocks in order.  At their boundaries we retain a DP
value for \emph{every} target position.  Let $v_i[p]$ be the least shortcut-graph cost of
reaching $(x_i,p)$.  Initially $v_0[p]\coloneqq p$ for each $p \in \fragmentcc{0}{n}$.
The exact shortcuts at source block $i$ can be applied individually.  The routing
shortcuts are too numerous to list: they have a label for every pair
$p\le q$.  We instead compute their full contribution at once:
\begin{equation}\label{eq:overview-routing-update}
    f_{i + 1}[q]\coloneqq \min_{p \in \fragmentcc{0}{q}}
       \{v_i[p] + \widehat r_i + \ed(\widehat C_i,Y\range{p}{q})\}.
\end{equation}

Hence, we may assume that we have the distances $v_i[p] + \widehat r_i$ to the vertices on the left boundary of the alignment graph of $\widehat C_i$ onto $Y$, and we want to compute the distances to the right boundary.

Note that the $(n + 1) \times (n + 1)$ matrix with entry $v_i[p] + \widehat r_i + \ed(\widehat C_i,Y\range{p}{q})$ in the $p$-th column of the $q$-th row is totally monotone:
\cref{sec:prelim-boundary} states this for the boundary-distance matrix even after adding a constant to each column.
Therefore, if we could get $\tOh(1)$-time oracle access to the values $\ed(\widehat C_i,Y\range{p}{q})$, we would be able to use the SMAWK algorithm (\cref{thm:prelim-smawk}) to compute all values of $f_{i+1}$ in $\tOh(n)$ time.

If we had a small number of centers and did not use center corrections but used $C_i$ instead of $\widehat C_i$, we could preprocess the alignment graph of $C_i$ onto $Y$ for each center $C_i$ (and reuse preprocessing for all $j$ with $C_j = C_i$) using the implicit DIST matrices (\cref{thm:dist}) and then get $\tO(1)$-time oracle access to all values $\ed(C_i,Y\range{p}{q})$.
In our case, all $\widehat C_i$ could be distinct so we would not be able to share preprocessing among many source blocks.
However, $\widehat C_i$ differs from $C_i$ only by at most a single edit, so we can still benefit from shared preprocessing.
More precisely, we can write
\[
 \widehat C_i=L_i U_i^\star R_i,\qquad |U_i^\star|\le1,
\]
where $L_i$ and $R_i$ are substrings of some sampled substring $B_p$ of $Y$.
Instead of propagating the distances directly through the whole $\widehat C_i$, we will first propagate them through $L_i$ using shared preprocessing of $B_p$, then trivially propagate them through a single column if $|U_i^\star| = 1$, and then propagate through $R_i$ again using shared preprocessing of $B_p$.

For each sampled segment $B_p$, build a balanced
binary tree over its character interval $\fragmentco{p}{\min\set{n, p + 2 W}}$.
Preprocess every binary tree node's substring against the whole target string $Y$ using
the implicit DIST representation (\cref{thm:dist}). 
Summing the oracle bounds over all depths gives $\tO(NW)$ time for preprocessing per sampled segment.
Every substring of $B_p$ is a concatenation of logarithmically many
node substrings.
See, e.g.,~\cite[Algorithm 11.7]{Tiskin13} for a similar procedure.

Therefore, by decomposing both $L_i$ and $R_i$ into $\Oh(\log N)$ binary tree pieces, we decompose the alignment graph of $\widehat C_i$ onto $Y$ into $\tOh(1)$ precomputed alignment graphs and a single alignment graph of $U_i^\star$ onto $Y$.
We process these subgraphs in order and propagate the distances through them.
Using the SMAWK algorithm (\cref{thm:prelim-smawk}) and oracle access to the DIST matrices of the preprocessed subgraphs, we process each subgraph in $\tOh(N)$ time.
Furthermore, we process the alignment subgraph of $U_i^\star$ onto $Y$ trivially in $\Oh(N)$ time.
In total, we compute all values of $f_{i+1}$ in $\tOh(N)$ time.

We then combine $f_{i+1}$ with exact shortcuts that we process directly and all horizontal and vertical edges to compute the values of $v_{i+1}$.
To process horizontal edges, we update each $v_{i+1}[q]$ with $v_i[q] + |U_i|$.
To process vertical edges, we update each $v_{i+1}[q]$ with $v_{i+1}[q-1]+1$, in increasing order for $q\in\irange{1}{n}$.
After processing all source blocks $U_i$, we return the answer $\widehat D\coloneqq v_R[n]$.

\subparagraph{Recovering an actual alignment.}
Backtracking the value $v_R[n]$ through the DP computation gives indices $y_0, \ldots, y_R$ such that the optimal path of length $\widehat D$ in the shortcut graph goes through vertex $(x_i, y_i)$ for each $i \in \fragmentcc{0}{R}$.
Run ordinary exact DP on the alignment graph between each adjacent pair $(x_i, y_i)$ and $(x_{i+1}, y_{i+1})$ of these vertices and concatenate the resulting
alignments.  Their total cost is no larger than $\widehat D$, and
their total construction time is
\[
 \Oh\left(\sum_i|U_i|(y_{i+1}-y_i+1)\right)=\Oh(NW)
\]
because each $|U_i|$ is bounded by $W$.

\subsection{The Bound on the Approximation Factor}
\label{sec:unit-overview-guarantee}

Consider an optimal alignment of $X$ onto $Y$. Recall its target substrings $V_i$, source-block costs $d_i$ with $\sum_i d_i=d$,
and starts $p_i$.
For a node $v$ of the binary tree in the exact shortcut search, write $D_v\coloneqq \sum_{i\in v}d_i$.
We prove that the shortcut graph with high probability contains a path of weight at most $3 d$ passing through each vertex $(x_i, p_i)$.
Consequently, $\widehat D \le 3d$ because it is the length of a shortest path without any restrictions.

\subparagraph{Lucky source blocks receive exact shortcuts.}
Call a source block $U_i$ \emph{useful} if $D_v < \rho_v$ holds for all nodes $v$ on the root-to-leaf path of $i$.
In particular, $d_i < r_i$ holds for the leaf.
Let $I$ be the set of useful source block indices.
Call a source block \emph{lucky} if every node on its root-to-leaf path samples at least one useful source block.
In particular, each lucky source block is useful.
All other source blocks are called \emph{unlucky}.
We claim that all lucky source blocks $U_i$ get their exact shortcut $(x_i, p_i) \to (x_{i+1}, p_{i+1})$.
At the root, the chosen useful source block $U_h$ compares against all substrings of $Y$ with distance below $r_h$.
Therefore, in particular we compare $U_h$ against $V_h$ at a cost $d_h < r_h$ and record the corresponding diagonal offset.
Suppose the parent $u$ of a node $v$ has the chosen anchor $\Delta = p_h - x_h$ of a
useful source block $h$, and $v$ samples a useful source block $j$.
The part of the fixed optimal alignment between their starting vertices lies
within $u$ and has cost at most $D_u$.  The offset bound gives
\[
 |(p_j-x_j)-\Delta|\le D_u < \rho_u.
\]
The second inequality holds because $U_h$ is useful.  Hence the child
searches $p_j$ and accepts the desired interval $V_j$ because $d_j < r_j$.
Induction reaches the leaf, whose only sampled source block is block $U_i$.
It outputs the exact shortcut for $(x_i, p_i) \to (x_{i+1}, p_{i+1})$.

\subparagraph{Unlucky source blocks have a bounded sum of corrected radii.}
Let $B$ consist of all unlucky source blocks.
That is, all source blocks whose tree path contains a node sampling
no useful source block.  At a node containing positive-radius blocks,
if at least an $\eta$ fraction of these blocks are useful, a logarithmic
multiple of $1/\eta$ independent draws hits one with high probability.
A union bound over all nodes shows that, with high probability, every
node that misses $I$ has at most an $\eta$ fraction of its
positive-radius blocks in $I$.  The assertion is automatic at zero-radius nodes.

Take the highest nodes that miss $I$.  Their source block sets are disjoint and
partition $B$.  Summing the preceding fraction bound over these nodes gives
\[
 |I\cap B|\le\eta b_+,\qquad
 b_+\coloneqq |\{i\in B:r_i>0\}|.
\]
Each radius $r_i$ is bounded by $W$ so the sum of useful unlucky radii is bounded by $\eta b_+ W$.
It remains to bound the sum of radii of all non-useful source blocks.
Consider the highest nodes with $D_v \ge \rho_v$.
Their source block sets are disjoint and partition the non-useful blocks.
Summing the inequalities $D_v \ge \rho_v$ over all of them, we get
\[\sum_{i \notin I} r_i \le d,\]
and thus, taking $\eta\coloneqq1/W$, we get
\[
 \sum_{i\in B}\widehat r_i
 =\sum_{i\in B}r_i-b_+ = \sum_{i \in B \setminus I} r_i + \sum_{i \in B \cap I} r_i - b_+
 \le d+\eta b_+ W-b_+ = d.
\]
The one-unit correction on each positive-radius unlucky block pays
for the total radius of the missed useful blocks.

\subparagraph{Putting everything together.}
Using exact shortcuts for the lucky source blocks and routing shortcuts for the unlucky ones, we obtain
\[
 \widehat D\le\sum_{i\notin B}d_i+
     \sum_{i\in B}(d_i+2\widehat r_i)
 =d+2\sum_{i\in B}\widehat r_i\le3d.
\]

\subsection{Running Time and the Resulting Theorem}
\label{sec:unit-overview-time}

We assume that all events that were claimed to happen with high probability indeed happen.

\subparagraph{Center selection and shared structures.}
For one sampled start, the DPs against its $\Oh(W)$-length segment cost
$\Oh(NW)$ over all source blocks.  Preprocessing the implicit DIST matrices (\cref{thm:dist}) at
all nodes of its decomposition binary tree costs $\tO(NW)$ as well.  Across the $\tO(N/k)$ samples, the center preprocessing costs $\tO(N^2W/k)$ time in total.

\subparagraph{The exact-shortcut search.}
Each exact-shortcut-search binary tree node samples $\tO(1/\eta)$ source blocks, and each has at most $k$
accepted starts.  It therefore has $\tO(k/\eta)$ anchors.
For one sampled child source block, the parent anchors generate intervals of
total width $\tO(k\rho_u/\eta)$.
The anchor computations using the implicit DIST matrices (\cref{thm:dist}) take time $\tO(k(\rho_u + W)W/\eta)$.
Multiplying by the $\tO(1/\eta)$ child samples and summing source
spans over the logarithmic-depth tree gives $\tO(NkW/\eta^2)$.
Here we use the fact that each radius $r_i$ is bounded by $|U_i|$ so the sum of radii at each binary tree level is bounded by $N$.
The root costs $\tO(NW/\eta)$ separately.
At each leaf there are at most $\Oh(kW)$ exact shortcuts for a total of $\Oh(NkW)$.

\subparagraph{Shortcut graph propagation and alignment reconstruction.}
Each of $R=\Oh(N/W)$ source blocks costs $\tO(N)$ to propagate through using the routing shortcuts for a total of $\tO(N^2 / W)$.
Applying the explicit shortcuts costs $\Oh(NkW)$ in total.
Alignment reconstruction takes $\Oh(NW)$ time.  The complete time complexity of the algorithm is therefore
\begin{equation}\label{eq:overview-unit-time}
 \tO\left(\frac{N^2W}{k}+\frac{NkW}{\eta^2}+\frac{N^2}{W}\right).
\end{equation}

Setting
\[
 W\coloneqq \lceil N^{1/6}\rceil,\qquad
 k\coloneqq \lceil N^{1/3}\rceil,\qquad \eta\coloneqq 1/W,
\]
gives a total time complexity of $\tO(N^{11/6})$.
Therefore, we have proved the first main theorem.

\unitmain*

\section{The Metric-Weighted Algorithm}
\label{sec:overview-weighted}

The metric-weighted algorithm follows a very similar roadmap to the unit-cost
algorithm of \cref{sec:overview-unit}.  Introducing weights
requires additional tools developed in~\cite{DKK26,MTW26}; we state the
ones we use where they are needed.
Our algorithm works in prescribed time complexity with high probability.
If the algorithm exceeds the desired time bound, we terminate the algorithm and return a trivial alignment of $X$ onto $Y$ that deletes all characters of $X$ and inserts all characters of $Y$.

Consider an arbitrary metric weight function $w \colon \overline\Sigma^2\to\mathbb R_{\ge0}$.
Put $d\coloneqq\ed^w(X,Y)$. If $X = Y$, return the trivial alignment of weight $0$; hence assume $d > 0$. For a character $a$,
define its \emph{gap cost} by $\gapcost(a)\coloneqq w(a,\gap)$.
For a string $S$, define its \emph{mass} as $\mass(S)\coloneqq\sum_{j=0}^{|S|-1}\gapcost(S[j])$.
Thus $\mass(S)$ is the cost of deleting/inserting the whole of $S$.
First run the algorithm of~\cite[Theorem 1.1]{DKK26} to find a $4$-approximation of $d$.
In $\tO(N^{8/5})$ time it returns a budget $D$ with
$d\le D$ on every outcome and $D\le4d$ with high probability.
Throughout the construction and its analysis, assume that this event holds.

\subsection{Reducing the Problem to Individual Well-Behaved Rectangles}

We first restrict the alignment graph to a small collection of
rectangles that together contain an optimal alignment.
Choose an integer $W \in \fragmentcc{2}{N}$ (eventually about $N^{1/7}$)
and an accuracy parameter $\eta \in (0, 1 / 16]$ (eventually about $N^{-2/21}$)
to control the approximation loss.

\begin{lemma}[Rectangles, source blocks, and target candidates,
  {\cite[Lemmas 5.1 and 6.1]{DKK26}}]
\label{lem:weighted-decomposition}
Given $D>0$, an integer $W \in \fragmentcc{2}{N}$, and a parameter $\eta \in (0, 1 / 16]$,
a deterministic $\tO(N/\eta)$-time algorithm constructs source--target
rectangles $(S_s,J_s)$ covering every vertex and edge of every alignment of $X$ onto $Y$
of cost at most $D$.  Their source projections have disjoint character
interiors, and we have
\[
 N_s\coloneqq |S_s|+|J_s|>0,\qquad \sum_sN_s=\Oh(N).
\]
For every rectangle with $|S_s|,|J_s|\ge W$, the algorithm also constructs:
\begin{enumerate}
\item a partition of $S_s$ into $R_s=\tO(N_s/W)$ nonempty \emph{source blocks}, each of length at most $W$;
\item a family $\mathcal Q_s$ of $M_s\coloneqq|\mathcal Q_s|$ substrings of $J_s$ called \emph{target candidates}, where
$M_s=\tO(N/(\eta^2W))$.
\end{enumerate}
Define the \emph{exact candidate graph} of such a rectangle to contain
its ordinary gap edges and, for every source block $U$ and target candidate
$Q\in\mathcal Q_s$, the corresponding certifying shortcut of cost exactly $\ed^w(U,Q)$.
For rectangles with a side shorter than $W$, use their full alignment graphs.
The union of these graphs, identifying shared boundary vertices, contains
for every alignment of $X$ onto $Y$ of cost $c\le D$ a route with the same
endpoints of cost at most
\begin{equation}\label{eq:overview-weighted-candidates}
 c+\Oh\!\left(\eta D\right).
\end{equation}
The algorithm outputs the rectangles, partitions, and target candidate endpoints,
without computing the shortcut costs.
\end{lemma}

We stress the fact that some rectangles may have empty source or target projections.
Rectangles with an empty side contain only gap edges.  Process every rectangle with $|S_s|<W$ or $|J_s| < W$ exactly by
ordinary DP on its alignment grid.
This takes $\Oh(\sum_s N_s W) = \Oh(NW)$ time.
The remaining rectangles have $|S_s|, |J_s|\ge W$.

If we could use exact candidate graphs for all remaining rectangles, their
combined error on an optimal alignment would be
$\Oh(\eta D)=\Oh(\eta d)$ by \cref{lem:weighted-decomposition}.
Thus any $\eta=o(1)$ would give a $(1+o(1))$-approximation.
Our goal is to approximate the distances in each exact candidate graph
by a factor $3-o(1)$, with enough saving below three to absorb that loss.

The cover simplifies the task.  We may search within the subgraph of the alignment graph
formed by the rectangles and process them from left to right.  At
each rectangle, approximate distances to its left boundary are already
available; our task is to propagate them to its right boundary.
We can now focus on one rectangle.

\subsection{Sampling Centers and Correcting One Substring}

Fix a remaining rectangle and abbreviate $S\coloneqq S_s$, $J\coloneqq J_s$,
keeping $N_s$ for its size.
The first item of \cref{lem:weighted-decomposition} partitions $S$
into $R\coloneqq R_s$ source blocks $U_i\coloneqq S\range{x_i}{x_{i+1}}$,
$i\in\fragmentco{0}{R}$, with $x_0=0$ and $x_R=|S|$.
Write $\mathcal Q\coloneqq\mathcal Q_s$ and $M\coloneqq M_s$.
If $M=0$, the exact candidate graph contains only gap edges.  Propagate
the boundary distances through it trivially in $\Oh(N_s)$ time and skip the
remaining steps for this rectangle.  Henceforth assume $M\ge1$.
As in \cref{sec:unit-overview-centers}, we seek a center for
each source block.

Under unit costs, a center longer than $2W$ cannot beat the empty
center.  We have no such length bound here: a center may contain
arbitrarily many symbols whose insertion costs are small.
The following lemma lets us omit most of these symbols while increasing
the distance to a source block by at most a factor of $1+\eta$.

\begin{lemma}[Short target candidate representations,
  {\cite[Lemma 7.2]{DKK26}}]
\label{lem:weighted-representatives}
After $\Oh(|J|)$-time preprocessing, a target candidate $Q$ has a
subsequence $Q'$ of length $\Oh(W/\eta)$, constructible in
$\tO(W/\eta)$ time, such that for every string $T$ of length at most $W$,
\begin{equation}\label{eq:overview-weighted-representation}
 \ed^w(T,Q')\le(1+\eta)\ed^w(T,Q).
\end{equation}
\end{lemma}

\subparagraph{Choose the centers.}
Choose an integer $k\ge1$ (eventually about $N^{8/21}$).
Sample uniformly and independently $\tO(1+M/k)$ target candidates
and construct their short representatives using \cref{lem:weighted-representatives}.
Add the empty string to the sampled representatives.
For each source block $U_i$, choose a center $C_i$ among these representatives
minimizing $\ed^w(U_i,C_i)$.  Retain an
optimal alignment to it, and define the radius
\begin{equation}\label{eq:overview-weighted-radii}
 r_i\coloneqq
 \frac{\ed^w(U_i,C_i)}{1+\eta}.
\end{equation}
The empty center candidate gives $r_i\le\mass(U_i)$, while
$\ed^w(U_i,C_i)=(1+\eta)r_i$ holds by definition.
For every \emph{sampled} target candidate $Q$,
\[
 r_i\le\frac{\ed^w(U_i,Q')}{1+\eta}
       \le\ed^w(U_i,Q).
\]
The same sampled-rank argument as in \cref{sec:unit-overview-centers}
therefore ensures with high probability that at most $k$ target candidates have distance
below $r_i$ from each source block.

\subparagraph{Replace one center substring.}
Choose $\tau \in (0, 1)$ (eventually a sufficiently big multiple of $\eta$) as the fraction of the source-to-center alignment
cost that we want to remove.
Partition $U_i$ into $q_i\coloneqq\lfloor1/\tau\rfloor$
consecutive pieces with lengths differing by at most one.
Split the recorded optimal alignment of $U_i$ onto $C_i$ into $q_i$ parts according to the pieces of $U_i$.

Choose a piece of maximum alignment cost $s_i$, and denote its source and center
substrings by $U_i^\star$ and $C_i^\star$ respectively.  Since the piece costs sum
to $\ed^w(U_i,C_i)$,
\[
 s_i\ge\frac{\ed^w(U_i,C_i)}{q_i}
       \ge\tau\ed^w(U_i,C_i).
\]
Write $C_i=L_iC_i^\star R_i$ and replace the selected substring by
its source partner:
\[
 \widehat C_i\coloneqq L_iU_i^\star R_i.
\]
The copied piece now matches for free, giving an alignment of $U_i$ onto
$\widehat C_i$ of cost $\widehat r_i \coloneqq\ed^w(U_i,C_i)-s_i$. Therefore,
\begin{equation}\label{eq:overview-weighted-repair}
 \ed^w(U_i,\widehat C_i)=\widehat r_i
 \le(1-\tau)\ed^w(U_i,C_i)
 =(1+\eta)(1-\tau)r_i.
\end{equation}
Moreover, $q_i\ge1/(2\tau)$, so the copied substring is short:
\[
 |U_i^\star|\le\lceil |U_i|/q_i\rceil =\Oh(\tau W+1).
\]
For every string $V$, corrected routing satisfies
\begin{equation}\label{eq:overview-weighted-routing}
 \ed^w(U_i,V)\le\widehat r_i+\ed^w(\widehat C_i,V)
 \le\ed^w(U_i,V)+2(1+\eta)(1-\tau)r_i.
\end{equation}

\subsection{Discovering the Cheap Shortcuts}

We again search for pairs of source blocks $U_i$ and target candidates $Q$
whose distance is small relative to $r_i$.
Because these shortcuts are not going to be exact anymore as in \cref{sec:unit-overview-search}, we instead call them \emph{cheap}.
We use the following consequence of the sparse-search analysis
of~\cite{DKK26}.  It bounds the total radius of the comparisons that the
search misses.

\begin{lemma}[Coverage by cheap shortcuts,
  {\cite[Lemmas 8.4--8.6]{DKK26}}]
\label{lem:weighted-sparse-coverage}
Given the source blocks and target candidates of the current rectangle,
radii $r_i \in [0, \mass(U_i)]$, and parameter $\eta \in (0, 1 / 16]$, a randomized
algorithm returns a set $H$ of block--candidate pairs and, for each
$(U_i,Q)\in H$, a label $\ell_i(Q)$ satisfying on every outcome
\[
 \ed^w(U_i,Q)\le\ell_i(Q)\le(1+\eta)\ed^w(U_i,Q).
\]
Fix an arbitrary path $P$ from the left boundary to the right boundary of the exact
candidate graph of cost $\operatorname{cost}(P)$.
Let $B(P)$ consist of the source block indices whose shortcut used by $P$
has its block--candidate pair outside $H$.
With high probability,
\begin{equation}\label{eq:weighted-missed-radii}
 \sum_{i\in B(P)}r_i\le(1+\Oh(\eta))\operatorname{cost}(P).
\end{equation}
If each source block has at most $k$ target candidates at distance below
its radius, then $|H|\le Rk$ and the running time on every
outcome is
\[
 \tO(RkW^2/\eta^4+N_sW/\eta+RM).
\]
\end{lemma}

Apply \cref{lem:weighted-sparse-coverage}. For every returned pair, add the corresponding
cheap shortcut with its returned label to the shortcut graph.

\subsection{Evaluating the Distance in the Shortcut Graph and Recovering the Alignment}

Our shortcut graph will consist of gap edges, routing shortcuts, and cheap shortcuts.
Cheap shortcuts are constructed explicitly while gap edges and routing shortcuts are processed in batches.
We proceed exactly the same way as in \cref{sec:unit-overview-evaluation}.
The only difference is that for a corrected center $\widehat C_i=L_iU_i^\star R_i$ we have $|U_i^\star|=\Oh(\tau W+1)$ rather than $|U_i^\star| \le 1$.
Therefore, we need to process $\Oh(\tau W + 1)$ columns of the alignment graph exactly rather than one.
Alignment reconstruction again works exactly the same way in time $\Oh(NW)$.

\subsection{The Bound on the Approximation Factor}

Consider a shortest path $P$ from an arbitrary vertex on the left boundary to an arbitrary vertex on the right boundary of the exact candidate graph of the current rectangle.
Let $d_\star$ be the cost of the path.
Split this path according to the partition $\set{U_i}_i$ and let $d_i$ be the cost of the $i$-th part.
We have $\sum_i d_i = d_\star$.

We prove that the shortcut graph of the current rectangle with high probability contains a path between the same pair of vertices of weight at most $(3 - 2 \tau + \Oh(\eta)) d_\star$.
It is sufficient to approximate the candidate shortcuts well because gap edges of the exact candidate graph are also present in the shortcut graph.
Put $B\coloneqq B(P)$ as in \cref{lem:weighted-sparse-coverage}.
That lemma gives
\[
 \sum_{i\in B}r_i\le(1+\Oh(\eta))d_\star.
\]

Use cheap shortcuts for the candidate comparisons outside $B$, routing shortcuts
for those in $B$, and keep the gap edges.  Blocks traversed using only gap
edges keep their original costs.  By \cref{eq:overview-weighted-routing},
we obtain a path of weight at most
\[ \sum_{i \notin B} (1 + \eta) d_i + \sum_{i \in B} (d_i + 2 (1 + \eta) (1 - \tau) r_i) \le (3 - 2 \tau + \Oh(\eta)) d_\star. \]
Center correction saves $2 \tau$ while the cheap edge search loses only $\Oh(\eta)$.
Taking the union bound gives this guarantee simultaneously
for all boundary pairs with high probability.

Now consider the whole optimal alignment of $X$ onto $Y$.
By \cref{lem:weighted-decomposition} and $d\le D\le4d$, the union of the
exact candidate graphs contains a route of cost at most $(1+\Oh(\eta))d$.
For rectangles with $|S_s| < W$ or $|J_s| < W$, we compute all distances exactly.
For all other rectangles, apply the preceding approximation of their
exact candidate graphs.  They concatenate into a path of total cost at most
\begin{equation}\label{eq:overview-weighted-budget}
 (3-2\tau+\Oh(\eta))(1+\Oh(\eta))d
 \le(3-2\tau+\Oh(\eta))d.
\end{equation}
Choose $\eta\coloneqq c_0\tau$ for a sufficiently small positive
constant $c_0$.  The bound is then at most
$(3-\tau)d\le3d$.

\subsection{Running Time and the Resulting Theorem}

We assume that all events that were claimed to happen with high probability indeed happen.
Choose the global parameters
\[
 W=\Theta(N^{1/7}),\qquad k=\Theta(N^{8/21}),\qquad
 \tau=\Theta(N^{-2/21}),\qquad \eta\coloneqq c_0\tau,
\]
where $c_0$ is the sufficiently small constant from the approximation analysis.

We first bound the work in one rectangle.
Recall that $R=\tO(N_s/W)$ and $M=\tO(N/(\eta^2W))$.
Center selection and shared preprocessing take
\[
 \tO\!\left(\frac{N_sW}{\eta}(1+M/k)\right)
 =\tO(N_sN^{19/21})
\]
time: there are $\tO(1+M/k)$ sampled representatives, each compared with
every source block and preprocessed against $J$.
By \cref{lem:weighted-sparse-coverage}, the cheap shortcut search takes
\begin{equation}\label{eq:overview-weighted-search-time}
 \tO(RkW^2/\eta^4+N_sW/\eta+RM)
 =\tO(N_sN^{19/21})
\end{equation}
time.

\subparagraph{Shortcut graph propagation and alignment reconstruction.}
Each source block costs $\tO(N_s(\tau W+1))$ time to propagate through its
corrected center.  Over all $R=\tO(N_s/W)$ blocks, boundary updates cost
\[
 \tO(N_sR(\tau W+1))=\tO(\tau N_s^2+N_s^2/W) =\tO(N_sN^{19/21})
\]
time.
Applying the explicit cheap shortcuts costs $\Oh(Rk)$. Alignment reconstruction costs
$\Oh(N_sW)$ in this rectangle.

Sum over all rectangles using $\sum_s N_s=\Oh(N)$.
Rectangles with $\min\{|S_s|,|J_s|\}<W$ cost $\Oh(NW)$ time in total.
The decomposition of \cref{lem:weighted-decomposition} costs
$\tO(N/\eta)$ time in total.  Alignment reconstruction adds
$\Oh(NW)$ time.  The initial $4$-approximation takes $\tOh(N^{8 / 5})$ time. The final time complexity is
\begin{equation}\label{eq:overview-weighted-time}
 \tO\!\left(N^{8/5}+NW+N/\eta+\sum_sN_sN^{19/21}\right)
 =\tO(N^{40/21}).
\end{equation}

This proves the second main theorem.

\weightedmain*

\section{Non-Metric Weight Functions}
\label{sec:nonmetric}

We now allow any finite-valued weight function $w$ with zero diagonal
and positive costs on unequal arguments, without requiring it to be a metric.
The distance is still
the minimum \emph{alignment} cost from \cref{sec:prelim-alignment}.
In particular, a column pays its specified cost directly: it cannot
substitute through intermediate symbols.

Our conclusion is that \emph{no approximation guarantee bounded in
terms of the string lengths is possible in truly subquadratic time}.
The obstruction is the
unrestricted size of the costs: a single cost $B$ can be made larger
than any proposed approximation allowance.
Our lower bounds work even if we relax a single condition of metric weights:
either allowing symmetric weight functions that do not satisfy the triangle inequality or non-symmetric weight functions satisfying the triangle inequality.
The alphabet size matters: three symbols suffice for symmetric hardness,
whereas symmetric binary weights admit a factor-three approximation.
For weights satisfying the triangle inequality without symmetry, two
symbols already suffice for hardness.

\subsection{Unconditional Lower Bound for Large Alphabets}\label{sec:oracle-lb}

\nonmetricoraclemain*

\begin{proof}
We first prove the lower bound for symmetric weights.
Fix an integer $h\ge1$ and set $B\coloneqq5h\alpha(7h+3)^2+1$.
Use pairwise distinct symbols
$\mathtt{\%},\mathtt{\$},a_0,\ldots,a_{h-1},b_0,\ldots,b_{h-1}$ and set
\[
 X\coloneqq\mathtt{\%}^{2h}\mathtt{\$}\,a_0a_1\cdots a_{h-1}\,
          \mathtt{\$}\mathtt{\%}^{2h},\qquad
 Y\coloneqq\mathtt{\$}b_0\mathtt{\$}b_1\cdots\mathtt{\$}
          b_{h-1}\mathtt{\$}.
\]
Thus $|X|=5h+2$, $|Y|=2h+1$, and $N=7h+3$.
The baseline weight function $w_0$ has zero diagonal.  For unequal arguments, give
cost one to every edit containing
$\mathtt{\%}$ and to insertions and deletions of symbols $a_i$.
Every other unequal pair has cost $B$.
For each $i\in\fragmentco{0}{h}$ and $j\in \fragmentco{0}{h}$, define the weight function $w_{ij}$ by
changing only the unordered pair $\{a_i,b_j\}$ from cost $B$ to cost
$1$.  All these functions are symmetric, but violate the triangle
inequality: for example,
$w_0(\mathtt{\$},\gap)=B>2=
 w_0(\mathtt{\$},\mathtt{\%})+w_0(\mathtt{\%},\gap)$.

\subparagraph{The baseline forces an expensive column.}
Suppose an alignment under $w_0$ has no column of cost $B$.
Each of the two dollars in $X$ must match a dollar in $Y$:
every other possible partner costs $B$, since $Y$ contains no percent
symbol.  The matched target dollars 
enclose at least one $b$-symbol.  Its partner lies between the source
dollars and is therefore an $a$-symbol or a gap.  Both possibilities
cost $B$, a contradiction.  Hence $\ed^{w_0}(X,Y)\ge B$.

\subparagraph{\boldmath One cheap pair gives an alignment of cost $5h$.}
Under $w_{ij}$, match the two source dollars to the dollars
surrounding $b_j$, match $a_i$ to $b_j$, and delete every other
$a$-symbol.
Use the leading and trailing percent blocks to match the remaining symbols of $Y$.
Each percent block costs $2h$, the $a$-block costs $h$, and the dollar matches are free.
Thus $\ed^{w_{ij}}(X,Y) \le 5h$.

\subparagraph{Any algorithm must inspect quadratically many weight table entries.}
A successful approximation is at least $B/\alpha(N)$ on the baseline
and at most $5h\alpha(N)$ on every alternative.  Our choice of $B$
ensures that
$5h\alpha(N)<B / \alpha(N)$.
Therefore, a valid algorithm must distinguish $w_0$ from any alternative $w_{ij}$.
Run the algorithm with the same random choices on $w_0$ and $w_{ij}$.
Let $Q_{ij}$ be the event that its baseline execution queries either $w(a_i, b_j)$ or $w(b_j, a_i)$.
Unless this happens, the executions coincide.
The probability of reporting the expensive case is at least $2/3$
on the baseline and at most $1/3$ on the alternative.  Their difference
is at most $\Pr[Q_{ij}]$, so $\Pr[Q_{ij}]\ge1/3$.
There are $h^2$ candidate pairs $(i, j)$.  By linearity of expectation, the baseline execution
makes $\Omega(h^2) = \Omega(N^2)$
queries in expectation.

\subparagraph{Imposing the triangle inequality instead.}
Keep the same strings and choose a directed baseline $w^{\to}_0$ with zero
diagonal and cost $B$ for every other pair except
\[
    w^{\to}_0(\mathtt{\%},z)=1\quad\text{for all $z\ne\mathtt{\%}$},\qquad
    w^{\to}_0(a_i,\gap)=1\quad\text{for all $i \in \fragmentco{0}{h}$}.
\]
For alternative $(i,j)$, let $w^{\to}_{ij}$ agree with $w^{\to}_0$ except that $w^{\to}_{ij}(a_i,b_j)=1$.
Triangle inequality in both cases is easy to verify.
These weight functions are not symmetric, since, e.g.,
$1 = w^{\to}_0(\mathtt{\%},\mathtt{\$}) \neq w^{\to}_0(\mathtt{\$},\mathtt{\%})=B$.
All weights used by an alignment of the constructed strings are
unchanged and the same argument follows.
\end{proof}

\subsection{Symmetric Hardness on Three Symbols}\label{sec:nonmetric-fixed}

The oracle bound uses a growing alphabet to hide a relevant table entry.
For a constant-size alphabet, we instead base hardness on the Orthogonal Vectors Hypothesis.

\begin{hypothesis}[Orthogonal Vectors Hypothesis~\cite{KW19}]
\label{hyp:ovh}
For every $\varepsilon>0$, there is a constant $c>0$ such that no
$\Oh(n^{2-\varepsilon})$-time algorithm decides whether two lists of
$n$ vectors in $\{0,1\}^{\ceil{c\log_2 n}}$ contain an orthogonal pair.
The randomized version of the hypothesis excludes algorithms succeeding with probability
at least $2/3$ on every input.
\end{hypothesis}

In particular, the Strong Exponential Time Hypothesis (SETH) implies OVH~\cite{KW19}.

\subparagraph{The intermediate problem.}
In \emph{subsequence matching with bounded gaps}, the input consists of
nonempty strings $P\in\{0,1\}^k$ and $T\in\{0,1\}^t$, together with
nonnegative integer upper bounds $\ell_0,\ldots,\ell_{k-2}$.
We ask whether there are positions
$0\le q_0<\cdots<q_{k-1}<t$ such that
\[
 T[q_i]=P[i]\quad(0\le i<k),\qquad
 q_{i+1}-q_i-1\le\ell_i\quad(0\le i<k-1).
\]
The following hardness result of Day, Kosche, Manea, and
Schmid~\cite{DKMS22} supplies the instances used by our reductions.

\begin{lemma}[{\cite[Theorem~3.3, Appendix D.3]{DKMS22}}]
\label{lem:bounded-gap-ov}
From two lists of $n$ binary vectors of dimension $d \ge 2$, one can
construct in $\Oh(nd)$ time a binary bounded-gap matching instance
$(P,T,\ell_0,\ldots,\ell_{k-2})$ of total size $\Oh(nd)$, with every
$\ell_i\le84$, that has a valid occurrence if and only if the lists
contain an orthogonal pair.
\end{lemma}

\subparagraph{One optional symbol encodes the gaps.}
Introduce a symbol $\star\notin\{0,1,\gap\}$.  It will either consume
one arbitrary target bit or be deleted, at cost one in either case.
Ordinary source bits will have to match identical target bits.

\begin{lemma}\label{lem:symmetric-gap-alignment}
Given a binary bounded-gap instance with $\ell_i\le84$ and an integer
$B>1$, one can construct in $\Oh(k+t)$ time strings $X,Y$ of total
length $\Oh(k+t)$ over $\{0,1,\star\}$ and a symmetric weight function $w$ with weights in
$\{0,1,B\}$ such that
\[
 \begin{aligned}
 \text{a valid occurrence exists}&\quad\Longrightarrow\quad
       \ed^w(X,Y)\le |X|,\\
 \text{no valid occurrence exists}&\quad\Longrightarrow\quad
       \ed^w(X,Y)\ge B.
 \end{aligned}
\]
\end{lemma}

\begin{proof}
Set
\[
 X\coloneqq\star^t P[0]\star^{\ell_0}P[1]\cdots
          \star^{\ell_{k-2}}P[k-1]\star^t,
 \qquad Y\coloneqq T.
\]
Give the weight function zero diagonal and set
\[
 w(\star,z)=w(z,\star)\coloneqq1
 \quad\text{for }z\in\{0,1,\gap\},
\]
with cost $B$ for every other unequal pair.
A valid occurrence gives an alignment by matching the ordinary source
bits to their occurrence positions, using the intervening stars for
skipped target bits, and deleting unused stars. The cost is 
$2t+\sum_{i=0}^{k-2}\ell_i\le |X|$.

Conversely, an alignment without a $B$-cost column has no target
insertions, and every ordinary source bit matches an identical target
bit.  At most $\ell_i$ target bits can be consumed by the stars between
$P[i]$ and $P[i+1]$.  The ordinary matches therefore give a valid
bounded-gap occurrence.  If no occurrence exists, every alignment
contains a $B$-cost column.  Finally,
$|X|+|Y|=3t+k+\sum_i\ell_i=\Oh(k+t)$.
\end{proof}

We are now ready to prove \cref{thm:nonmetric-fixed}.

\nonmetricfixedmain*

\begin{proof}
Suppose the claimed approximation algorithm exists for some
$\delta \in (0, 1)$.  Choose the constant $c$ in
\cref{hyp:ovh} for $\varepsilon\coloneqq\delta/2$.
Given two lists of $n$ vectors of dimension $\ceil{c\log_2n}$, apply
\cref{lem:bounded-gap-ov,lem:symmetric-gap-alignment}.  The strings do
not depend on $B$ and have total length $N=\Oh(n\log n)$.
Compute $\alpha(N)$ and choose
\[
 B\coloneqq(\alpha(N)^2+1)N.
\]
Construct the constant-size weight function table, call the approximation algorithm,
and answer yes exactly when $\alpha(N)\widehat D<B$.
On a successful approximation, a yes-instance satisfies
\[
 \alpha(N)\widehat D\le\alpha(N)^2|X|<B,
\]
whereas a no-instance satisfies
$\alpha(N)\widehat D\ge\ed^w(X,Y)\ge B$.
The running time is
\[
 \Oh\bigl((n\log n)^{2-\delta}\bigr)+n^{1+o(1)}
   =\Oh(n^{2-\delta/2}),
\]
contradicting the Orthogonal Vectors Hypothesis.
\end{proof}

\subsection{An Approximation for Symmetric Weights over a Binary Alphabet}
\label{sec:binary-symmetric}

Ideally, we would want to prove \cref{thm:nonmetric-fixed} for binary alphabets.
One may ask whether the alphabet size $3$ is a limitation of our specific construction, and a more intricate argument may show that arbitrary approximation for binary alphabets is also OVH-hard.
Interestingly, it is not the case.
In this section we show that weighted edit distance over a binary alphabet with a symmetric weight function can be reduced to the metric case and thus can be approximated in truly subquadratic time.

For a binary alphabet, taking the metric closure of the weight function changes the edit
distance by an exactly computable, nonnegative amount.  This sign is
what allows a multiplicative approximation to transfer.
Write $\#_s(S)$ for the number of occurrences of a symbol $s$ in a string $S$.

\begin{lemma}\label{lem:binary-metric-correction}
Let $w$ be a symmetric weight function on $\{0,1,\gap\}$, and put
\[
 a\coloneqq w(0,\gap),\qquad b\coloneqq w(1,\gap),\qquad
 c\coloneqq w(0,1).
\]
Define the metric weight function $v$ by its zero diagonal and the costs
\[
 \begin{aligned}
 v(0,\gap) = v(\gap, 0)&= a'\coloneqq\min\{a,b+c\},\\
 v(1,\gap) = v(\gap, 1)&= b'\coloneqq\min\{b,a+c\},\\
 v(0,1) = v(1,0)&= c'\coloneqq\min\{c,a+b\}.
 \end{aligned}
\]
For arbitrary strings $X,Y\in\{0,1\}^*$ and $s\in\{0,1\}$, set
\[
 E_s\coloneqq\max\{0,\#_s(X)-|Y|,\#_s(Y)-|X|\}.
\]
Then
\begin{equation}\label{eq:binary-metric-correction}
 \ed^w(X,Y)=\ed^v(X,Y)+(a-a')E_0+(b-b')E_1.
\end{equation}
Moreover, given an arbitrary alignment $\mathcal A$ of $X$ onto $Y$, one can
construct in $\Oh(|X| + |Y|)$ time another alignment $\mathcal A'$ of $X$ onto $Y$ satisfying
\[
 \operatorname{cost}_w(\mathcal A')
 \le\operatorname{cost}_v(\mathcal A)+(a-a')E_0+(b-b')E_1.
\]
\end{lemma}

\begin{proof}
The metric $v$ is the shortest-path metric of the three-vertex undirected graph
with edge lengths $a,b,c$.  These lengths are positive, and at most
one of them can exceed the sum of the other two.

If $c>a+b$, replacing every substitution by a deletion followed by an
insertion shows that capping $c$ at $a+b$ preserves edit distance.
The resulting weights already form a metric, and both correction
terms vanish.  It remains, up to exchanging the two symbols, to
consider $a>b+c$.  In this case, only the zero-gap weight changes:
$a'=b+c$, $b'=b$, and $c'=c$.

Every alignment has at least $E_0$ zero-gap columns.
We prove that some $v$-optimal alignment has exactly $E_0$ such columns.
Represent an alignment by the monotone matching of its nongap source
and target positions.
Call these positions \emph{selected}. All other positions are unmatched.  Among the
$v$-optimal matchings, choose one minimizing the total number of
unmatched zeros.

Suppose an unmatched source zero coexists with a matched source one.
Choose a nearest matched one, so that every matched source symbol
between these positions is zero.  Replace that matched one by the
unmatched zero in the subsequence of selected source positions, and pair the two
selected subsequences in order.  Only one symbol of the selected
source subsequence changes, from one to zero.  Its pairing weight
increases by at most $c$, while its gap weight decreases by
$a'-b=c$.  The new matching is therefore still optimal and has one
fewer unmatched zero, a contradiction.  By symmetry the same argument
applies on the target side.

Consequently, if a source zero is unmatched, every matched source
symbol is zero.  If any target position with symbol $z$ is also
unmatched, add it and an unmatched source zero to the selected
subsequences and pair them in order.  Every previously matched target
symbol still pairs with a zero, so its cost is unchanged.  The added
pair decreases the total cost by
\[
 a'+v(\gap,z)-v(0,z)
 =\begin{cases}2a',&z=0,\\2b,&z=1,\end{cases}
\]
which is strictly positive.  Thus every target position is matched,
and exactly $\#_0(X)-|Y|$ source zeros are unmatched.  The reverse
argument handles unmatched target zeros.  If neither side has an
unmatched zero, $E_0=0$.  These cases prove that the chosen optimum
has exactly $E_0$ zero-gap columns.

The $w$-cost of any alignment is its $v$-cost plus $a-a'$ for each
zero-gap column.  The lower bound of $E_0$ on their number and the
$v$-optimal alignment attaining this bound give
$\ed^w(X,Y)=\ed^v(X,Y)+(a-a')E_0$, as required.
Here $(b - b')E_1 = 0$.

To make the conversion constructive, start from any alignment.
If $w$ is metric, keep it; if $c>a+b$, replace its substitutions
by deletions and insertions as above.  For $a>b+c$, the exchanges
above do not increase the $v$-cost of an arbitrary alignment either.
On each string, scan left to right, maintaining a stack of matched
ones.  Whenever an unmatched zero has such a one to its left,
exchange it with the nearest one, at the top of the stack.  Repeat
with a right-to-left scan.  Afterwards, an unmatched zero cannot
coexist with a matched one on the same side.
Next, add pairs of unmatched zeros until at most one side has any
left.  Pair its remaining unmatched zeros with unmatched positions
on the other side until either supply is exhausted.  The side
supplying zeros has only zeros selected, so each addition decreases
the $v$-cost by the quantity displayed above.  The preceding argument
then shows that exactly $E_0$ zero-gap columns remain.  Hence the
resulting $w$-cost is at most the original $v$-cost plus
$(a-a')E_0$.  The case $b>a+c$ is symmetric.
Maintain only the selected-position flags during these operations
and pair the selected positions in order once at the end.
The algorithm takes $\Oh(N)$ time in total.
\end{proof}

We are now ready to prove \cref{thm:binary-symmetric}.

\binarysymmetricmain*

\begin{proof}
Run \cref{thm:weighted} with the weight function $v$ of \cref{lem:binary-metric-correction} and let $\mathcal A_v$ be the
alignment it returns, of $v$-cost $\widehat D_v$.  Use
\cref{lem:binary-metric-correction} to convert it in $\Oh(N)$ time
into an alignment $\mathcal A$ of $w$-cost
$\widehat D\coloneqq\operatorname{cost}_w(\mathcal A)$, and return
$\mathcal A$.  Write $K\coloneqq(a-a')E_0+(b-b')E_1\ge0$.
On the success event of that algorithm, \cref{lem:binary-metric-correction}
gives
\[
    \widehat D \le\widehat D_v+K\le3\ed^v(X,Y)+K\le3\ed^w(X,Y).
\]
\end{proof}

\subsection{Binary Alphabet Hardness with the Triangle Inequality}
\label{sec:binary-quasimetric}

Without symmetry, we show that even for binary alphabets the problem is hard.
One of the two symbols can act as the gap symbol $\star$ in $X$.

\begin{lemma}\label{lem:binary-directed-gap-alignment}
Given a binary bounded-gap instance with $\ell_i\le84$ and an integer
$B>1$, one can construct in $\Oh(k+t)$ time strings $X,Y\in\{0,1\}^*$
of total length $\Oh(k+t)$ and a weight function $w^\to$ satisfying
the triangle inequality, with all weights in $\{0,1,B\}$, such that
\[
 \begin{aligned}
 \text{a valid occurrence exists}&\quad\Longrightarrow\quad
       \ed^{w^\to}(X,Y)\le |X|,\\
 \text{no valid occurrence exists}&\quad\Longrightarrow\quad
       \ed^{w^\to}(X,Y)\ge B.
 \end{aligned}
\]
\end{lemma}

\begin{proof}
Give $w^\to$ zero diagonal and set
\[
 w^\to(1,0)=w^\to(1,\gap)\coloneqq1,
\]
with cost $B$ for every other unequal pair. It is easy to verify that this weight function satisfies the triangle inequality.

An alignment avoiding cost $B$ has no target insertions, every source
zero matches a target zero, and a source one may match either target
bit or be deleted.
For $a\in\{0,1\}$, define
\[
 C_a\coloneqq 0^4\,1\,a\,(1-a)\,1\,0^3.
\]
Set
\[
 \begin{aligned}
 Y&\coloneqq C_{T[0]}\,1\,C_{T[1]}\,1\cdots1\,C_{T[t-1]},\\
 X&\coloneqq 1^{|Y|}\,C_{P[0]}\,1^{1+12\ell_0}\,C_{P[1]}\cdots
        1^{1+12\ell_{k-2}}\,C_{P[k-1]}\,1^{|Y|}.
 \end{aligned}
\]
The strings satisfy $|X|+|Y|=\Oh(k+t)$.
We first show that every source codeword in an alignment avoiding $B$
must match an identical target codeword in full.  Its initial $0^4$ block
must match four consecutive target zeros.  The only target zero runs
of length at least four are the initial runs of target codewords,
and they have length exactly four.  Hence the source codeword starts
at a target codeword boundary.
Its terminal $0^3$ block must match a later target zero run of length at
least three. The
same target codeword's terminal $0^3$ is the only possible run:
reaching even the first three zeros of the next codeword would force an insertion.
Thus, the source codeword consumes precisely the eleven characters
of that target codeword, with no deletions.  All eight source zeros have to
match zeros.  As there are exactly eight target zeros too, no source
one can substitute to zero.  The codewords must therefore be identical.

Between target codewords that skip $r$ original text symbols, there
are exactly $1+12r$ target characters.  The intervening source block
of $1+12\ell_i$ ones can consume them if and only if $r\le\ell_i$.
The outside blocks can consume any remaining target prefix or suffix.
Consequently, a valid bounded-gap occurrence gives an alignment with
no $B$-cost column, whose cost is at most $|X|$: each source position
costs at most one and there are no insertions.  Conversely, every
alignment without a $B$-cost column determines identical matched
codewords with the required gap bounds, hence a valid occurrence.
\end{proof}

We are now ready to prove \cref{thm:binary-quasimetric-hardness}.

\binaryquasimetrichardmain*

\begin{proof}
Suppose the claimed algorithm exists for some $\delta \in (0, 1)$. Choose $c$ in \cref{hyp:ovh} for
$\varepsilon\coloneqq\delta/2$ and reduce the resulting OV instance
using \cref{lem:bounded-gap-ov,lem:binary-directed-gap-alignment}.
The strings have total length $N=\Oh(n\log n)$ and do not depend on
$B$.  Compute $\alpha(N)$, set $B\coloneqq(\alpha(N)^2+1)N$, and answer
yes exactly when the approximation satisfies $\alpha(N)\widehat D<B$.
On a successful approximation, a yes-instance gives
$\alpha(N)\widehat D\le\alpha(N)^2|X|<B$, and a no-instance gives
$\alpha(N)\widehat D\ge\ed^{w^\to}(X,Y)\ge B$.
The total running time
is $\Oh(n^{2-\delta/2})$. This contradicts the Orthogonal Vectors Hypothesis.
\end{proof}

In this case, the choice of the binary alphabet is easy to justify:
weighted edit distance over a unary alphabet with an arbitrary weight function is trivially solvable exactly in linear time.

\paragraph*{Acknowledgements and AI disclosure.}
The author thanks Tomasz Kociumaka for helpful discussions.

The project began as an attempt by the author to understand the metric $(3+\eps)$-approximation algorithm of Mader, Tavasoli, and Wang~\cite{MTW26}. The author interactively prompted an LLM to understand and try to simplify that algorithm. Inspired by short target candidate representations similar to \cref{lem:weighted-representatives} (an idea originally attributed to Kuszmaul~\cite{Kusz19}), the author suggested the idea of center corrections. Given this idea, an LLM autonomously produced a factor-$3$ approximation algorithm for metric edit distance in $\tOh(N^{2-1/32})$ time.

The initial proof was considerably more complicated than necessary. The author distilled its main ideas and applied them to obtain a much simpler factor-$3$ algorithm for unit-cost edit distance, inspired by Andoni~\cite{Andoni20} and presented in \cref{sec:overview-unit}. Based on this unit-cost proof, an LLM produced a simpler and faster metric-weighted algorithm. Through further interactive prompting, this algorithm was simplified into the form presented in \cref{sec:overview-weighted}. During this process, the author chose to switch and follow the proof of Das, Kipouridis, and Kociumaka~\cite{DKK26} instead of that of Mader, Tavasoli, and Wang~\cite{MTW26}, because it offered more opportunities for black-box reductions.

The author subsequently proved the unconditional oracle lower bound of \cref{sec:oracle-lb} and a fine-grained lower bound for alphabets of size four based on the regular-expression membership lower bounds of Backurs and Indyk~\cite{BI16}. (Frankly, the proofs of both results are easy exercises.) Based on these findings, an LLM proposed the more convenient intermediate problem of Day, Kosche, Manea, and Schmid~\cite{DKMS22} and used it to reduce the alphabet size to two in the quasimetric case and three in the symmetric case.
These are the results presented in \cref{sec:nonmetric-fixed,sec:binary-quasimetric}.
Puzzled by this apparent discrepancy in alphabet sizes, the author investigated the symmetric binary-alphabet case and proved the upper bound of \cref{sec:binary-symmetric}.

The author wrote a preliminary draft of this paper in plain English, omitting lower-level proof details. An LLM expanded these notes into the first full draft, which the author subsequently edited extensively. The author takes full responsibility for the contents of this manuscript.

\begingroup
\small
\raggedright
\bibliography{refs}

@misc{DKK26,
  author = {Das, Debarati and Kipouridis, Evangelos and Kociumaka, Tomasz},
  title = {{Metric Weighted Edit Distance: $(3+\varepsilon)$-Approximation in $\widetilde{\mathcal{O}}_{\varepsilon}(N^{1.6})$ Time}},
  year = {2026},
  month = sep,
  eprint = {2609.20796v1}
}

@misc{MTW26,
  author = {Mader, Ethan and Tavasoli, Borna and Wang, Jihan},
  title = {{A Strongly Subquadratic $(3+\varepsilon)$-Approximation for Weighted Edit Distance over Arbitrary Metrics}},
  year = {2026},
  month = sep,
  eprint = {2609.14873v1}
}

@inproceedings{GRS21,
  author = {Goldenberg, Elazar and Rubinstein, Aviad and Saha, Barna},
  title = {{Does Preprocessing Help in Fast Sequence Comparisons?}},
  booktitle = {Proceedings of STOC},
  pages = {657--670},
  year = {2020},
  doi = {10.1145/3357713.3384300},
  eprint = {2108.09115}
}

@inproceedings{MR26,
  author = {Mao, Xiao and Rubinstein, Aviad},
  title = {{Approximation Schemes for Edit Distance and LCS in Quasi-Strongly Subquadratic Time}},
  booktitle = {Proceedings of STOC},
  pages = {744--754},
  year = {2026},
  doi = {10.1145/3798129.3800789}
}

@misc{Andoni20,
  author = {Andoni, Alexandr},
  title = {{Simple Constant-Factor Approximation to Edit Distance}},
  year = {2020},
  month = aug,
  url = {https://www.cs.columbia.edu/~andoni/papers/edit/simple.pdf}
}

@misc{Andoni18,
  author = {Andoni, Alexandr},
  title = {{Simpler Constant-Factor Approximation to Edit Distance Problems}},
  year = {2018},
  month = oct,
  url = {https://www.cs.columbia.edu/~andoni/papers/edit/edit3.pdf}
}

@article{SMAWK,
  author = {Aggarwal, Alok and Klawe, Maria M. and Moran, Shlomo and Shor, Peter and Wilber, Robert},
  title = {Geometric applications of a matrix-searching algorithm},
  journal = {Algorithmica},
  volume = {2},
  number = {1},
  pages = {195--208},
  year = {1987},
  doi = {10.1007/BF01840359}
}

@article{HLLW,
  author = {Hermelin, Danny and Landau, Gad M. and Landau, Shir and Weimann, Oren},
  title = {Unified compression-based acceleration of edit-distance computation},
  journal = {Algorithmica},
  volume = {65},
  number = {2},
  pages = {339--353},
  year = {2013},
  doi = {10.1007/s00453-011-9590-6}
}

@misc{Tiskin13,
  author = {Tiskin, Alexander},
  title = {{Semi-local string comparison: Algorithmic techniques and applications}},
  year = {2013},
  month = nov,
  doi = {10.1007/s11786-007-0033-3},
  eprint = {0707.3619}
}

@inproceedings{CDGKS,
  author = {Chakraborty, Diptarka and Das, Debarati and Goldenberg, Elazar and Kouck{\'y}, Michal and Saks, Michael E.},
  title = {Approximating edit distance within constant factor in truly sub-quadratic time},
  booktitle = {Proceedings of FOCS},
  pages = {979--990},
  year = {2018},
  url = {https://iuuk.mff.cuni.cz/~koucky/papers/approxEdit-v1.pdf}
}

@article{WF,
  author = {Wagner, Robert A. and Fischer, Michael J.},
  title = {The string-to-string correction problem},
  journal = {Journal of the ACM},
  volume = {21},
  number = {1},
  pages = {168--173},
  year = {1974},
  doi = {10.1145/321796.321811}
}

@article{Ukk85,
  author = {Ukkonen, Esko},
  title = {Algorithms for approximate string matching},
  journal = {Information and Control},
  volume = {64},
  number = {1--3},
  pages = {100--118},
  year = {1985},
  doi = {10.1016/S0019-9958(85)80046-2}
}

@article{LV88,
  author = {Landau, Gad M. and Vishkin, Uzi},
  title = {Fast string matching with $k$ differences},
  journal = {Journal of Computer and System Sciences},
  volume = {37},
  number = {1},
  pages = {63--78},
  year = {1988},
  doi = {10.1016/0022-0000(88)90045-1}
}

@article{MP80,
  author = {Masek, William J. and Paterson, Michael S.},
  title = {A faster algorithm computing string edit distances},
  journal = {Journal of Computer and System Sciences},
  volume = {20},
  number = {1},
  pages = {18--31},
  year = {1980},
  doi = {10.1016/0022-0000(80)90002-1}
}

@inproceedings{BI15,
  author = {Backurs, Arturs and Indyk, Piotr},
  title = {{Edit Distance Cannot Be Computed in Strongly Subquadratic Time (unless SETH is false)}},
  booktitle = {Proceedings of STOC},
  pages = {51--58},
  year = {2015},
  doi = {10.1145/2746539.2746612},
  eprint = {1412.0348}
}

@inproceedings{AKO10,
  author = {Andoni, Alexandr and Krauthgamer, Robert and Onak, Krzysztof},
  title = {Polylogarithmic approximation for edit distance and the asymmetric query complexity},
  booktitle = {Proceedings of FOCS},
  pages = {377--386},
  year = {2010},
  doi = {10.1109/FOCS.2010.43}
}

@inproceedings{AN20,
  author = {Andoni, Alexandr and Nosatzki, Negev Shekel},
  title = {Edit distance in near-linear time: it's a constant factor},
  booktitle = {Proceedings of FOCS},
  pages = {990--1001},
  year = {2020},
  doi = {10.1109/FOCS46700.2020.00096}
}

@misc{FMQTW26,
  author = {Fox, Emily and Mader, Ethan and Quanrud, Kent and Tavasoli, Borna and Wang, Jihan},
  title = {{A Faster 3-Approximation for Sublinear Edit Distances}},
  year = {2026},
  month = jun,
  url = {https://ethan-mader.com/assets/files/sublinear_edit_distance.pdf}
}

@inproceedings{Kusz19,
  author = {Kuszmaul, William},
  title = {{Dynamic Time Warping in Strongly Subquadratic Time: Algorithms for the Low-Distance Regime and Approximate Evaluation}},
  booktitle = {46th International Colloquium on Automata, Languages, and Programming (ICALP 2019)},
  series = {Leibniz International Proceedings in Informatics},
  volume = {132},
  pages = {80:1--80:15},
  year = {2019},
  doi = {10.4230/LIPIcs.ICALP.2019.80},
  eprint = {1904.09690}
}

@article{Huang88,
  author = {Huang, Xiaoqiu},
  title = {A lower bound for the edit-distance problem under an arbitrary cost function},
  journal = {Information Processing Letters},
  volume = {27},
  number = {6},
  pages = {319--321},
  year = {1988},
  doi = {10.1016/0020-0190(88)90220-7}
}

@inproceedings{CKW23,
  author = {Cassis, Alejandro and Kociumaka, Tomasz and Wellnitz, Philip},
  title = {{Optimal Algorithms for Bounded Weighted Edit Distance}},
  booktitle = {Proceedings of FOCS},
  pages = {2177--2187},
  year = {2023},
  doi = {10.1109/FOCS57990.2023.00135}
}

@inproceedings{KW19,
  author = {Kane, Daniel M. and Williams, Richard Ryan},
  title = {The Orthogonal Vectors Conjecture for Branching Programs and Formulas},
  booktitle = {Proceedings of ITCS},
  series = {Leibniz International Proceedings in Informatics},
  volume = {124},
  pages = {48:1--48:15},
  year = {2019},
  doi = {10.4230/LIPIcs.ITCS.2019.48}
}

@misc{BGK25LCS,
  author = {Boneh, Itai and Golan, Shay and Kraus, Matan},
  title = {Deterministic Longest Common Subsequence Approximation in Near-Linear Time},
  year = {2025},
  eprint = {2507.22486}
}

@article{CLZ03,
  author = {Crochemore, Maxime and Landau, Gad M. and Ziv-Ukelson, Michal},
  title = {A Subquadratic Sequence Alignment Algorithm for Unrestricted Scoring Matrices},
  journal = {SIAM Journal on Computing},
  volume = {32},
  number = {6},
  pages = {1654--1673},
  year = {2003},
  doi = {10.1137/S0097539702402007}
}

@article{Schmidt98,
  author = {Schmidt, Jeanette P.},
  title = {All Highest Scoring Paths in Weighted Grid Graphs and Their Application to Finding All Approximate Repeats in Strings},
  journal = {SIAM Journal on Computing},
  volume = {27},
  number = {4},
  pages = {972--992},
  year = {1998},
  doi = {10.1137/S0097539795288489}
}

@inproceedings{DKMS22,
  author = {Day, Joel D. and Kosche, Maria and Manea, Florin and Schmid, Markus L.},
  title = {Subsequences with Gap Constraints: Complexity Bounds for Matching and Analysis Problems},
  booktitle = {Proceedings of ISAAC},
  series = {Leibniz International Proceedings in Informatics},
  volume = {248},
  pages = {64:1--64:18},
  year = {2022},
  doi = {10.4230/LIPIcs.ISAAC.2022.64},
  eprint = {2206.13896}
}

@inproceedings{BI16,
  author = {Backurs, Arturs and Indyk, Piotr},
  title = {{Which Regular Expression Patterns Are Hard to Match?}},
  booktitle = {Proceedings of FOCS},
  pages = {457--466},
  year = {2016},
  doi = {10.1109/FOCS.2016.56},
  eprint = {1511.07070}
}
\endgroup

\clearpage

\end{document}